%% file: CM_Security_Allocation.tex
\documentclass[runningheads]{llncs}
\usepackage[T1]{fontenc}
\usepackage{cite}
\usepackage[hidelinks]{hyperref}
\usepackage{amsmath,amssymb,amsfonts}
\usepackage{algpseudocode,algorithm}
\usepackage{graphicx}
\usepackage{dsfont}
\usepackage{textcomp}
\usepackage{subcaption}
\usepackage{mathtools}
\usepackage{xcolor}
\usepackage{tabularx}
\usepackage{academicons}
\usepackage{multirow}
\usepackage{diagbox}
\usepackage[colorinlistoftodos]{todonotes}

\input{deflatex3.tex}

\newtheorem{Assumption}{Assumption}
\newtheorem{Theorem}{Theorem}
\newtheorem{Definition}{Definition}
\newtheorem{Remark}{Remark}
\usepackage{graphicx}
\begin{document}
\title{Centrality-based
Security Allocation in
Networked Control Systems}
\titlerunning{Centrality-based
Security Allocation}
%
\author{Anh Tung Nguyen\orcidID{0000-0001-9316-233X} \and Andreas Hertzberg  \and \\ André M. H. Teixeira\orcidID{0000-0001-5491-4068}}
\authorrunning{Anh Tung Nguyen et al.}
%
\institute{Department of Information Technology, Uppsala University, \\ PO Box 337, SE-75105, Uppsala, Sweden \\
\email{anh.tung.nguyen@it.uu.se} \\
\email{andreas.hertzberg.0811@student.uu.se} \\
\email{andre.teixeira@it.uu.se}
}
\maketitle              
\begin{abstract}
This paper addresses the security allocation problem within networked control systems, which consist of multiple interconnected control systems under the influence of two opposing agents: a defender and a malicious adversary. The adversary aims to maximize the worst-case attack impact on system performance while remaining undetected by launching stealthy data injection attacks on one or several interconnected control systems. Conversely, the defender's objective is to allocate security resources to detect and mitigate these worst-case attacks. A novel centrality-based approach is proposed to guide the allocation of security resources to the most connected or influential subsystems within the network. The methodology involves comparing the worst-case attack impact for both the optimal and centrality-based security allocation solutions. The results demonstrate that the centrality measure approach enables significantly faster allocation of security resources with acceptable levels of performance loss compared to the optimal solution, making it suitable for large-scale networks. The proposed method is validated through numerical examples using Erdős–Rényi graphs.

\end{abstract}
\section{Introduction}
\label{sec:intro}
Networked control systems are integral to modern infrastructure, which are prominently featured in power grids, transportation networks, and water distribution systems. These systems leverage open and widely-used information and communication technologies, including the public Internet and wireless communication \cite{teixeira2015secure,falliere2011w32,kshetri2017hacking}. However, using unprotected wireless communication channels may expose them to cyber threats, which can lead to severe financial losses and societal disruptions. Notable examples include the Iranian industrial control system, which suffered from the Stuxnet malware in 2010 \cite{falliere2011w32}, and the Ukrainian power grid, which was compromised by the Industroyer malware in 2016 \cite{kshetri2017hacking}. These incidents underscore the critical importance of bolstering security measures within control systems to protect against such potentially devastating cyber attacks.

The issue of cybersecurity has garnered significant attention from researchers within the Information Technology community. Solutions in this area are critical for maintaining the core principles of the Confidentiality, Integrity, and Availability (CIA) triad for data and IT services. These principles are defined as follows: 1) Confidentiality: ensures that data stored, transmitted, and processed remains private and accessible only to authorized individuals; 2) Integrity: ensures that data is reliable and cannot be modified by unauthorized means; and 3) Availability: ensures that the system provides timely access to data when requested. Among these, integrity is particularly emphasized, as any breach directly impacts system performance \cite{teixeira2015secure}. In exploring potential threats to control systems, Zhang et al. \cite{zhang2021stealthy} introduced the concept of controlled invariant subspace, derived from geometric control theory, which remains independent of system outputs and nonlinear functions, to propose a covert attack on certain classes of nonlinear systems. Addressing nonlinear systems presents considerable challenges, requiring attackers to gain deep system knowledge and intercept real-time transmitted data. Rather than focusing on nonlinear system models, Park et al. \cite{park2019stealthy} examined potential threats in linear systems with uncertainties. By thoroughly analyzing system dynamics and conventional anomaly detection mechanisms, these sophisticated stealthy data injection attacks significantly disrupt control systems without being detected. This highlights the need for more advanced detection methods. Least squares-based fusion techniques are widely employed to identify and mitigate the harmful effects of false data injection attacks on state estimation \cite{li2023secure,li2024secure}. Physical watermarking, initially designed to detect replay attacks, has been adapted to signal the presence of false data injection attacks \cite{mo2013detecting}. Recent advancements have further reduced the average detection delay time \cite{naha2023quickest}. However, physical watermarking can increase control costs and reduce actuator lifespan in control systems. Multiplicative watermarking \cite{teixeira2019optimal,ferrari2020switching,gallo2021design} offers an effective solution to these challenges. Additionally, reallocating defensive resources can help mitigate the adverse effects of attackers \cite{anand2022risk}, enhancing the system's resilience against stealthy false data injection attacks. Nevertheless, it is important to acknowledge that contemporary intelligent adversaries can learn and adapt to these mitigation strategies, refining their attack methods to maintain stealth. Therefore, this paper emphasizes the need for a comprehensive investigation of networked control systems under the influence of stealthy false data injection attacks.

Upon review of the above existing studies \cite{li2023secure,li2024secure,zhang2021stealthy,teixeira2019optimal,ferrari2020switching,gallo2021design,naha2023quickest,mo2013detecting,park2019stealthy,anand2022risk}, it becomes evident that the literature has primarily focused on secure estimation and control from the viewpoint of either the defender or the adversary.
However, it is important to acknowledge that both parties encounter similar challenges. While the defender confronts resource limitations in countering malicious activities, the adversary also faces energy constraints during attack execution. Therefore, it is highly pertinent to address this matter within a comprehensive framework that encompasses both the defender and the adversary. Fortunately, the security allocation problem between the defender and the malicious adversary fits well within the framework of game theory, which is
one of the most efficient frameworks addressing the challenge of optimal decision-making among non-cooperative parties \cite{bacsar1998dynamic}.
In an attempt to apply game theory to deal with the security allocation problem, the authors in \cite{li2018false} seek the Stackelberg equilibrium by formulating the optimization problems in \cite[Theorem 3.3]{li2018false} where all the attack scenarios are considered in their constraints. To solve those optimization problems in \cite[Theorem 3.3]{li2018false}, the authors use a linear mapping from action space to the game payoff (see more in \cite[Section III.B]{li2018false}), removing the complexity of computing game payoff.
In the same line with \cite{li2018false}, the authors in \cite{yuan2019stackelberg} also seek the Stackelberg equilibrium by considering all the attack scenarios in \cite[Theorem 1]{yuan2019stackelberg}. The study in \cite{shukla2022robust} has a similar game setup to ours where the defender and the adversary have discrete action spaces. They also conduct the traditional backward induction to find the Stackelberg equilibrium \cite[Algorithm 1]{shukla2022robust}. To be able to proceed with the algorithm, the authors in \cite{shukla2022robust} need to investigate all the possible action scenarios shown in \cite{shukla2022robust}, resulting in a very high computational cost \cite[Section V]{shukla2022robust}. Although not considering a Stackelberg game setting, the authors in \cite{purposeadver,milovsevic2023strategic} tackle an 
analogous problem. The authors in \cite{purposeadver} find the Bayesian game equilibrium through solving an optimization problem \cite[Section VI.A]{purposeadver} where the game payoff matrix is employed in its constraint. Clearly, building up the game payoff matrix enumerates all the action scenarios \cite[Tables 1 and 2]{purposeadver}. Although the problem in \cite{milovsevic2023strategic} is not clearly formulated as a game between a defender and an adversary, the authors discuss the mutual formulation between theirs and a zero-sum game setting in \cite[Section IV.B]{milovsevic2023strategic}. The authors in \cite{milovsevic2023strategic} find the optimal action for the operator, who has the same role as the defender in our work, through solving an optimization problem in \cite[Section III.C]{milovsevic2023strategic} where they consider all the possible actions in discrete action spaces. Further, many other concepts of games
describing networked systems subjected to cyber attacks such
as matrix games \cite{nguyen2022single, purposeadver, shukla2022robust, nguyen2022zero, nguyen2023optimal}, dynamic games \cite{gupta2016dynamic}, and stochastic games \cite{miao2018hybrid} have been studied.
In summary, the consideration of all the action scenarios remains in recent studies  \cite{shukla2022robust,purposeadver,li2018false,yuan2019stackelberg,milovsevic2023strategic}. Particularly, the authors in \cite{li2018false,yuan2019stackelberg,milovsevic2023strategic} deal with security problems by solving a single large optimization problem that addresses all the action scenarios. As a consequence, the computation in those existing studies is considerably heavy. We aim to address such an issue in this study.

In this paper, we consider a continuous-time networked control system, associated with an undirected connected graph, involving two strategic agents: a defender and a malicious adversary. The system consists of multiple interconnected subsystems, referred as to vertices. The aim of the adversary is to maximally disrupt the global performance of the network without being detected by the defender. 
Meanwhile, the defender selects several monitor vertices to measure their outputs with the purpose of alleviating the attack impact. To assist the defender in allocating the defense resources, we adopt centrality measures including degree, betweenness, and closeness centrality measures to seek the most potential monitor vertices. This approach enables us to solve less complicated optimization problems, massively alleviating the computational cost compared to finding the optimal solution.

\textbf{Notation:} the set of real positive numbers is denoted as $\Rbb_+$ ; $\Rbb^n$ and $\Rbb^{n \times m}$ stand for sets of real $n$-dimensional vectors and $n$-row $m$-column matrices, respectively.
A vector with the $i$-th element set to one and the other elements set to zero is denoted
$e_i \in \Rbb^n$.
A positive definite matrix $A$ and a positive semi-definite matrix $B$ are denoted as $A \succ 0$ and $B \succeq 0$, respectively. The notations $A \prec 0$ and $B \preceq 0$ stand for $-A \succ 0$ and $-B \succeq 0$, respectively.
The space of square-integrable  functions is defined as $\Lc_{2} \triangleq \bigl\{f: \Rbb_{+} \rightarrow \Rbb ~|~ \norm{f}^2_{\Lc_2 [0,\infty]} < \infty \bigr\}$ and the extended space be defined as $\Lc_{2e} \triangleq \bigl\{ f: \Rbb_{+} \rightarrow \Rbb ~|~ \norm{f}^2_{\Lc_2 [0,T]} < \infty,~ \forall~ 0 < T < \infty \bigr\} $.
The notation $\norm{x}^2_{\Lc_2}$  is used  as shorthand for the norm $\norm{x}_{\Lc_2 [0,T]}^2 \triangleq \frac{1}{T}\int_{0}^{T} \norm{x(t)}_2^2~\text{d}t$ if the time horizon $[0,T]$ is clear from the context.
Let $\Gc \triangleq (\Vc, \Ec, A)$ be {an undirected} graph with the set of $N$ vertices $\Vc = \{1, 2,...,N\}$, the set of edges $\Ec \subseteq \Vc \times \Vc $, and the  adjacency matrix $A = [a_{ij}]$.
For any $(i,j) \in \Ec, ~i\neq j$, the element of the adjacency matrix $a_{ij}$ is positive, and with $(i,j) \notin \Ec$ or $i = j$, $a_{ij} = 0$. 
The degree of vertex $i$ is denoted as 
$\Delta_i =  \sum_{j=1}^{n} a_{ij}$ and the degree matrix of graph $\Gc$ is defined as 
$\Delta = {\bf diag}\big(\Delta_1, \Delta_2,\dots, \Delta_N\big)$, where ${\bf diag}$ stands for a diagonal matrix. 
The Laplacian matrix is defined as $L = [\ell_{ij}] = \Delta - A$.
Further, $\Gc$ is called an undirected connected graph if and only if matrix $A$ is symmetric and the algebraic multiplicity of zero as an eigenvalue of $L$ is one.
The set of all neighbours of vertex $i$ is denoted as $\Nc_i = \{j \in \Vc~|~ (i,j) \in \Ec \}$.

\section{Problem Description}
\label{sec:prob}
In this section, we present the mathematical description of a networked control system under cyber attacks. Then, the purposes and resources of the adversary and the defender are introduced.
\subsection{Networked control systems under false data injection attacks}
Consider a networked control system associated with an undirected connected graph $\Gc \triangleq (\Vc, \, \Ec, \, A)$ with $N$ vertices
where every vertex $i$ is described by the one-dimensional state-space model:
\begin{equation}
\dot{x}_{i}(t) = u_{i}(t), \, i \in \{1, 2, \ldots, N\},
\label{eq:statei}
\end{equation}
where ${x}_{i}(t) \in \mathbb{R}$ represents the state variable of vertex $i$, and $u_{i}(t) \in \mathbb{R}$ denotes the control input for the same vertex. The state of the entire network will be denoted as $x(t) \triangleq [x_{1}(t), x_{2}(t), \ldots, x_{N}(t)]^\top \in \mathbb{R}^N$. The control law for each vertex is:
\begin{equation}
u_{i}(t) = \sum_{j \, \in \, \mathcal{N}_{i}} a_{ji} \, (x_{j}(t)-x_{i}(t)),~ \forall \, i \in \Vc.
\label{eq:kontrolllaw}
\end{equation}
By applying the control law \eqref{eq:kontrolllaw} to the state-space model \eqref{eq:statei} for the entire network, one obtains the following closed-loop model \eqref{eq:close} without attacks:

\begin{equation}
\dot{x}(t) = -Lx(t),
\label{eq:close}
\end{equation} 
where $L$ is a Laplacian matrix representing the graph $\Gc$.

When the system is in the presence of a malicious adversary, 
we assume that the adversary selects a set of vertices to launch false data injection attacks (says $\zeta(t)$). Let the set of attack vertices be denoted as $\mathcal{A} \triangleq \{a_1, a_2, \dots, a_{n_{a}}\}$ for a given attack budget $n_{a} = |\mathcal{A}|$. It is worth noting that there is no benefit of choosing fewer attack vertices than the budget and therefore we ignore the case of $|\mathcal{A}| <  n_{a}$. 
The $N$-dimensional binary vector $B(\Ac) \in \{0,\,1\}^N$ stands for the chosen attack vertices where its $i$-th element $B(\Ac)_i = 1$ if $i$-th vertex is attacked and $B(\Ac)_i = 0$ otherwise. 
Then, the networked control system \eqref{eq:close} under false data injection attacks can be rewritten as follows:
\begin{equation}
    \dot{x}(t) = -Lx(t) + B(\Ac)\zeta(t).
\label{eq:closeatt}
\end{equation}
Let us make use of the following assumption.
\begin{Assumption}
    The networked control system \eqref{eq:closeatt} is at its equilibrium $x_e = 0$ before being affected by attack signals.
\end{Assumption}
In the following parts, the network performance and the purposes of the adversary and the defender will be described in more detail.

\subsection{Network performance and monitoring}
In the networked control system, the overall performance of the system is the collective performance of all vertices in the system. Consequently, the cost function for the system performance can be denoted as:
\begin{align}
    J &= \sum_{i \, \in \, \mathcal{V}} \| y_{i} \|_{\mathcal{L}_2}^2, \label{eq:netper}\\ 
    y_{i}(t) &= e_{i}^\top x(t)  , \ \forall \, i \in \mathcal{V}.
    \label{eq:outputrho}
\end{align}
The security challenge against the malicious adversary stems from the assumed budget limitation within the system, making it necessary to select a subset of vertices to monitor for the residual signal since it is not feasible to monitor all vertices. This security resource constraint enforces the defender chooses a subset of the vertex set $\mathcal{V}$ to serve as the monitor set $\mathcal{M} = \{m_{1}, m_{2}, \ldots, m_{n_s}\}$ for a given sensor budget $n_{s} = |\mathcal{M}|$. The outputs of the monitor vertices can be written as follows:
\begin{equation}
y_{m_k}(t) = e_{m_k}^\top x(t), ~\forall \, m_k \in \Mc.
\label{eq:monitoreq}
\end{equation}
By monitoring outputs \eqref{eq:monitoreq}, the defender can notify the presence of the adversary if the energy of the monitor outputs crosses a given alarm threshold $\delta$, i.e.,
\begin{align}
    \norm{y_{m_k}}_{\Lc_2}^2 > \delta.
\end{align}
\begin{Remark}
    In this paper, we assume that the alarm threshold $\delta$ is given by the operator, which cannot be altered by the defender. Therefore, the problem of choosing the alarm threshold is not considered in this study. Instead, the defender is allowed to choose vertices, which have their given corresponding alarm thresholds, to monitor their outputs with the purpose of detecting malicious activities.
\end{Remark}

\subsection{The purposes of adversary and defender}
In the following, we first present the purpose of the adversary. Then, the purpose of the defender defender is introduced.

\textbf{The purpose of the malicious adversary} is to maximize the negative impact on the system performance, i.e., the adversary seeks to maximize the cost function $J$ in \eqref{eq:netper}. Furthermore, there are strong arguments as to why the malicious adversary simultaneously would seek to remain stealthy to the defender, (see the discussion in \cite[Section II-E]{purposeadver}). One can assume that the defender would quickly mitigate the attack when detected by shutting down the maliciously affected communication channels. This means there will be no payoff for the adversary, additionally one can assume that the malicious adversary has invested resources to obtain system knowledge and to infiltrate the system, which will most likely be lost when detected. Therefore, to avoid losing the investment the adversary has made, this paper considers that the malicious adversary conducts stealthy data injection attacks, which will be defined in the following. 

Let us consider the networked control system under cyber attacks \eqref{eq:closeatt} and the monitor outputs \eqref{eq:monitoreq}.
The false data injection attack $\zeta(t)$ in \eqref{eq:closeatt} is called a stealthy false data injection attack if, and only if, the monitor outputs $\|y_{m_{k}}\|_{\mathcal{L}_2}^2 < \delta$ for all $m_k \in \mathcal{M}$. If at least one monitor vertex  $m_k \in \mathcal{M}$ does not satisfy this, the attack from the malicious adversary is detected.

Additionally, a realistic assumption is that the adversary has finite resources not only in the number of attack vertices but also in the energy of the attack signals. In more detail, 
let us denote a positive number $A_{e}$ as the maximum attack energy of the attack signal $\zeta(t)$ 
for a given time horizon $[0, T]$, i.e. 
\begin{align}
    \| \zeta \|_{\mathcal{L}_2[0, T]}^2 \leq A_{e}.
\end{align}

\textbf{The purpose of the defender} is to allocate security resources given the potential existence of a malicious adversary that seeks to negatively impact the system's performance. However, a realistic assumption is that the defender is unable to foresee the action of the adversary. Therefore, the defender needs to formulate a defense strategy without knowing the attack strategy. An approach that has not been considered in existing literature, is the utilization of centrality measures to 
select monitor vertices in the networked control system. This approach will be utilized and presented more thoroughly in the following sections.

\section{Risk analysis and evaluation}
\label{sec:risk}
Given the networked control system under cyber attacks outlined in \eqref{eq:closeatt}, the network performance \eqref{eq:netper}, and the monitor outputs \eqref{eq:monitoreq}, the attack set $\Ac$, and the monitor set $\Mc$, 
the worst-case attack impact (WCAI) on the network performance is formulated as follows: 
\begin{equation}
\begin{aligned}
J(\mathcal{A},\mathcal{M}) = \sup_{\zeta \, \in \, \mathcal{L}_{2e}}&~~ \sum_{i \, \in \, \mathcal{V}} \, \| y_{i} \|_{\mathcal{L}_2}^2  \\
\text{s.t.}~~&~~  \|y_{m_{k}}\|_{\mathcal{L}_2}^2 < \delta, \, \forall \,  m_k \in \mathcal{M}, 
\\
&~~  \| \zeta \|_{\mathcal{L}_2}^2 \leq A_{e},  
\\
&~~ \eqref{eq:closeatt},\eqref{eq:outputrho},\eqref{eq:monitoreq},~ x(0) = 0.
\label{eq:Jwsup}
\end{aligned}
\end{equation}
The following theorem presents the boundedness of the non-convex optimization problem \eqref{eq:Jwsup} and its computation.
\begin{Theorem}
The worst-case impact of stealthy FDI attacks is always bounded and computed by the following semi-definite programming (SDP) problem.
\begin{align}
J(\Ac,\,\Mc) = &\min_{P = P^\top 
 \succeq 0, \, \beta \,>\, 0 ,\, \gamma_{m_{k}} \,>\, 0}\ ~A_{e}\, \beta + \delta\sum_{m_{k} \in \, \mathcal{M}} \gamma_{m_{k}} \label{eq:JwSDP} \\
\text{s.t.}~&~ \ba{cc}
-L P - PL ~~~& PB(\Ac) \\
B(\Ac)^\top P & -\beta
\ea + \sum_{i \, \in \, \mathcal{V}} 
\ba{c}
e_{i}^\top \\
0
\ea
\ba{cc}
e_{i} & 0
\ea \non \\
&~
- \sum_{m_{k} \, \in \, \mathcal{M}}\gamma_{m_{k}}
\ba{c}
e_{m_{k}}^\top \\
0
\ea
\ba{cc}
e_{m_{k}} & 0
\ea \non \preceq 0.
    \end{align} \QET
\end{Theorem}
\begin{proof}
    The networked control system \eqref{eq:closeatt} is stable due to the fact that matrix $-L$ is negative semi-definite. Meanwhile, the attack input $\zeta$ has bounded energy. As a result, the output performance $y_i$ also has bounded energy for all $i \in \Vc$, leading to the boundedness of \eqref{eq:Jwsup}.
    Next, the computation of \eqref{eq:Jwsup} is shown. Let us consider the dual form of \eqref{eq:Jwsup} described in the following:
    \begin{align}
        \inf_{\gamma_k > 0, \, \beta > 0} &\bigg[ \sup_{\zeta \, \in \, \, \mathcal{L}_{2e} } \sum_{i \, \in \, \mathcal{V}} \| y_{i} \|_{\mathcal{L}_2}^2 + \sum_{m_k \in \Mc} \gamma_k \big( \delta - \|y_{m_{k}}\|_{\mathcal{L}_2}^2 \big) + \beta \big( A_e - \| \zeta \|_{\mathcal{L}_2}^2 \big) 
        \bigg] \label{dual_form} \\
        \text{s.t.}~~~&~
        \eqref{eq:closeatt},\eqref{eq:outputrho},\eqref{eq:monitoreq},~ x(0) = 0. \non
    \end{align}
    The dual form \eqref{dual_form} is bounded if, and only if, 
    \begin{align}
        \sum_{i \, \in \, \mathcal{V}} \| y_{i} \|_{\mathcal{L}_2}^2 - \sum_{m_k \in \Mc} \gamma_k \|y_{m_{k}}\|_{\mathcal{L}_2}^2 - \beta \| \zeta \|_{\mathcal{L}_2}^2 \leq 0, ~ \forall \, \zeta \, \in \, \, \mathcal{L}_{2e}, ~ x(0) = 0.
    \end{align}
    As a result, the dual form can be rewritten as follows:
    \begin{align}
        J(\mathcal{A},\mathcal{M}) = \inf_{\gamma_k \,>\, 0, \, \beta \,>\, 0}& ~~ \delta \sum_{m_k \in \Mc} \gamma_k + A_e \beta  \label{J_inf} \\
        \text{s.t.}~~&~\sum_{i \, \in \, \mathcal{V}} \| y_{i} \|_{\mathcal{L}_2}^2 - \sum_{m_k \in \Mc} \gamma_k \|y_{m_{k}}\|_{\mathcal{L}_2}^2 - \beta \| \zeta \|_{\mathcal{L}_2}^2 \leq 0, \non \\
        &~\eqref{eq:closeatt},\eqref{eq:outputrho},\eqref{eq:monitoreq},~ x(0) = 0. \non 
    \end{align}
    Recalling the key results in dissipative system theory for linear systems \cite{trentelman1991dissipation} with a non-negative storage function $V(x) = x(t)^\top P x(t)$ where $P = P^\top \succeq 0$ and a supply rate $s(\cdot,\cdot) =  \sum_{m_k \in \Mc} \gamma_k \|y_{m_{k}}\|_{\mathcal{L}_2}^2 + \beta \| \zeta \|_{\mathcal{L}_2}^2 -\sum_{i \, \in \, \mathcal{V}} \| y_{i} \|_{\mathcal{L}_2}^2 $, the optimization problem \eqref{J_inf} can be computed by the SDP \eqref{eq:JwSDP}. \QEDB 
\end{proof}





Given that the adversary seeks to maximize the worst-case attack impact on the system, the defender selects a monitor set to minimize the worst-case attack impact. In practice, the defender seldom foresees when the adversary attacks the system. Thus, the defender should make a decision on their defense strategy before the adversary does. As a result, the optimal security allocation can be formulated as the following minimax optimization problem:
\begin{align}
    \min_{\mathcal{M},\,|\Mc| = n_s} ~ \bigg[ \max_{\mathcal{A},\, |\Ac|  = n_a} J(\mathcal{A},\mathcal{M}) \bigg].
\end{align}
The above optimal security allocation can be rewritten as follows:
\begin{equation}
\begin{aligned}
\min_{\mathcal{M},\,|\Mc| = n_s} \; & \: Q \\
 \text{s.t.}~~~&  J(\mathcal{A},\mathcal{M}) \leq Q, \: \forall \, \mathcal{A},\, |\Ac|  = n_a.
\label{eq:Q}
\end{aligned}
\end{equation}
In the next section, we discuss how to compute \eqref{eq:Q} and its approximate solution based on centrality measures.

\section{Security Allocation}
\label{sec:method}
This section outlines the methodology employed in this paper, encompassing all methods and approaches utilized to derive the results. The novel approach of selecting the monitoring set $\mathcal{M}$ utilizing centrality measures, and the approach of assessing monitoring sets. 

\subsection{Optimal security allocation}
Inspired by our previous work \cite[Proposition 1]{Tung2024security}, the defender is able to find the optimal security allocation by solving the combinatorial problem \eqref{eq:Q}. The following theorem states the computation of the optimization problem \eqref{eq:Q}.
\begin{Theorem}
    For each attack set $\Ac$, let us denote a tuple variable $(\bar z_\Ac,P_\Ac) \in \Rbb^{N} \times \Sbb^N$ correspondingly. Denote an $N$-dimensional binary vector $z_\Mc \in \{0,\,1\}^N$ as a representation of the monitor set $\Mc$ where $m$-entry of $z_\Mc$ being equal to $1$ indicates that $m$-th node belongs to $\Mc$. Suppose that the adversary finds an attack set $\Ac$ such that it maximizes the worst-case impact of stealthy FDI attacks \eqref{eq:Jwsup}. Then, the optimal security allocation \eqref{eq:Q} is determined by $z_\Mc^\star$ which is the solution to the following mixed-integer SDP problem:
    \begin{align} 
        &\min_{\bar z_\mathcal{A} \, \in \, \Rbb^N, \,z_{\mathcal{M}} \in \{0,1\}^N, \, P_\mathcal{A} = P^\top_\Ac \succeq 0, \, \beta \, > \, 0, \, Q \, > \, 0}~~ Q \label{mixedintegersdpoptimal} \\
        &\text{s.t.}~~~
        \textbf{1}_{N}^\top z_{\mathcal{M}} = n_{s},~
        \delta \textbf{1}_N^\top \bar z_\mathcal{A} + A_{e}\, \beta \leq Q, ~
        0 \leq \bar z_\mathcal{A} \leq \tilde M \, z_{\mathcal{M}}, ~ |\Ac|  = n_a, \non \\
        &
        \ba{cc}
        - L^\top P_\mathcal{A} - P_\mathcal{A} L ~& ~~P_\mathcal{A} B({\Ac}) \\
        B({\Ac})^\top P_\mathcal{A} & -\beta_{\Ac}
        \ea + \textbf{diag} \left( \ba{c} I \\ 0 \ea \right)  
        - \textbf{diag} \left( \ba{c} \bar z_\mathcal{A} \\ 0 \ea \right) \leq 0,~
        \forall \, \mathcal{A},  \non 
\end{align}
where $\textbf{1}_N$ is an $N$-dimensional all-one vector, $n_s$ is the sensor budget, $n_a$ is the attack budget, $\delta$ is the given alarm threshold, $A_e$ is the maximum attack energy, and $\tilde M$ is a given large positive number, also called a ``big M"  \cite{bigM}. 
\end{Theorem}
\begin{proof}
    The proof directly follows \cite[Proposition 1]{Tung2024security}.
    \QEDB
\end{proof}

\subsection{Centrality-based security allocation} 
\label{sec:centralitymeasures}
The methodology of selecting monitor sets based on the centrality measures including degree, betweenness, and closeness centrality relies on the idea that the most connected or influential vertices in the graph would be the most suitable monitor vertices. 
Therefore, we create all possible monitor sets for each centrality measure given the monitor budget $n_{s}$. Then, we derive the total centrality for each monitor set by summing up the calculated centrality measure for each vertex. Finally, we select the monitor set with the highest total centrality according to the specific centrality measure. Given the sensor budget, this ensures that the vertices with the highest centrality are selected as the monitor set. This process, while straightforward, often results in duplicate solutions where different monitor sets may have the same total centrality. For example, if the sensor budget is two and there are three vertices with the highest and identical centrality, there must be multiple monitor sets with the same highest total centrality. In these cases, all monitor sets are evaluated by solving \eqref{eq:JwSDP} and selecting the monitor set with the smallest WCAI. This method is applied to all three centrality measures which are degree centrality, closeness centrality, and betweenness centrality measures. Let us make use of the following definition of the above-mentioned centrality measures.
\begin{Definition}
    [Centrality measures \cite{Golbeck2013Chapter3}] Given a graph $\Gc \triangleq (\Vc,\, \Ec, \, A)$ and $d(i,\,j)$ as the shortest path from node $i$ to node $j$, the degree centrality measure of node $i \in \Vc$, denoted as $C_D(i)$, is the total number of connections between node $i$ and its neighbors, i.e.,
    \begin{align}
        C_D(i) \triangleq \Delta_i = \sum_{j \, \in \, \Nc_i} a_{ij}.
    \end{align}
    The closeness centrality measure of node $i \in \Vc$, denoted as $C_C(i)$, is the inverse-average shortest path length from the vertex to all the other vertices, i.e.,
    \begin{align}
        C_C(i) \triangleq \frac{N-1}{\sum_{j \neq i}d(i,\, j)}.
    \end{align}
    The betweenness centrality measure of node $i \in \Vc$, denoted as $C_B(i)$, is the number of shortest paths that pass through vertex $i$ of interest among all shortest paths between pairs of vertices in the graph, i.e.,
    \begin{align}
        C_B(i) \triangleq \sum_{s \neq i \neq t} \frac{\sigma_{st}(i)}{\sigma_{st}},
    \end{align}
    where $\sigma_{st}$ is the total number of shortest paths from vertex $s$ to $t$ and $\sigma_{st}(i)$ is the number of shortest paths from vertex $s$ to $t$ that pass through vertex $i$.
\end{Definition}

Moreover, a combined centrality measure approach is considered where the three monitor sets (obtained from the respective centrality metrics) are first compared to each other, and the one with the smallest WCAI is selected and then compared with the optimal solution. The combined method allows us to ascertain if utilizing only one of these centrality measures would generally give a larger gap to the optimal solution than using all three and always picking the best one.

Let us denote the set of monitoring vertices founded by degree centrality measure as $z_\Mc^d$, closeness centrality measure as $z_\Mc^c$, betweenness centrality measure as $z_\Mc^b$, combined centrality measure as $z_\Mc^o$.
After we obtain the set of monitoring vertices that has the highest centrality measures, denoted $z_\Mc^\dagger \in \{0,\,1\}^N$ where $\dagger \in \{d,\,c,\,b,\,o\}$, we can compute the worst-case impact of stealthy FDI attacks by replacing the known value $z_\Mc^\dagger$ with the binary variable $z_\Mc$ in \eqref{mixedintegersdpoptimal}. It is worth noting that the optimization problem \eqref{mixedintegersdpoptimal} with $z_\Mc^\dagger$ is an SDP problem and can be solved much faster than the mixed-integer version presented in \eqref{mixedintegersdpoptimal} with $z_\Mc$ as a binary variable.

\section{Results}
\label{sec:result}
This section describes the empirical results of centrality measure-based security allocation in networked control systems. In the first part, we show experimental results on Erdős–Rényi random graphs where an edge is included to connect two vertices with a probability of 0.5. In the second part, we present how the centrality measure assists the security allocation in an IEEE benchmark for power systems (IEEE 14-bus system).

\subsection{Erdős–Rényi random graphs}
All the experiments were performed using Matlab 2023b
with YALMIP 2021 \cite{lofberg2004yalmip} and MOSEK solver on a personal computer with 2.9-GHz, 8-core Intel i7-10700 processor and 16 GB of RAM. Due to the limited hardware, the size of networks is restricted from $10$ to $20$. It is worth noting that not only the size of the network influence the complexity of the security allocation problem \eqref{mixedintegersdpoptimal}, but the attack scenarios also significantly contribute to its complexity. More specifically, given the attack budget $n_a$, the number of attack scenarios is computed by $\binom{N}{n_a}$, significantly scaling up with the network size and the attack budget. Therefore, we limit our experiment to the network size $N = \{10, \, 12, \, 14, \, 16, \, 18, \, 20 \}$, the attack budget $n_a = \{ 1 , \, 2 \}$, and the monitor budget $n_s = 1$. For each network size, we did the experiment with $30$ Erdős–Rényi random graphs.

The experimental results are shown in Figures~\ref{fig:wcai10_14}-\ref{fig:wcai16_20} for the relative gaps in WCAI and Figures~\ref{fig:time10_14}-\ref{fig:time16_20} for the relative gaps in solving time. As seen in Figure~\ref{fig:wcai10_14}, monitor sets chosen based on combined centrality and betweenness centrality measures perform very well where their median and 75 percentile values are kept under 10\%. In particular, 75 percentile values of the combined centrality measure in case of $2$ attack budget stay at zero. Regarding solving time in Figure~\ref{fig:time10_14}, security allocation based on combined centrality measures reduces around 70\% solving time while the other centrality measures decrease approximately 90\% solving time compared with solving the optimal value. When we increase the size of the network up to 20 depicted in Figure~\ref{fig:wcai16_20}, the relative gaps in WCAI increase in all centrality measures, particularly degree and closeness centrality measures. However, the numbers for combined centrality and betweenness centrality measures are still kept under 15\%. The relative gaps in solving time are almost the same in the case of smaller networks (see Figure~\ref{fig:time16_20}). 
\begin{figure}[!t]
    \centering
    \begin{subfigure}{0.49\textwidth}
        \includegraphics[width=\textwidth]{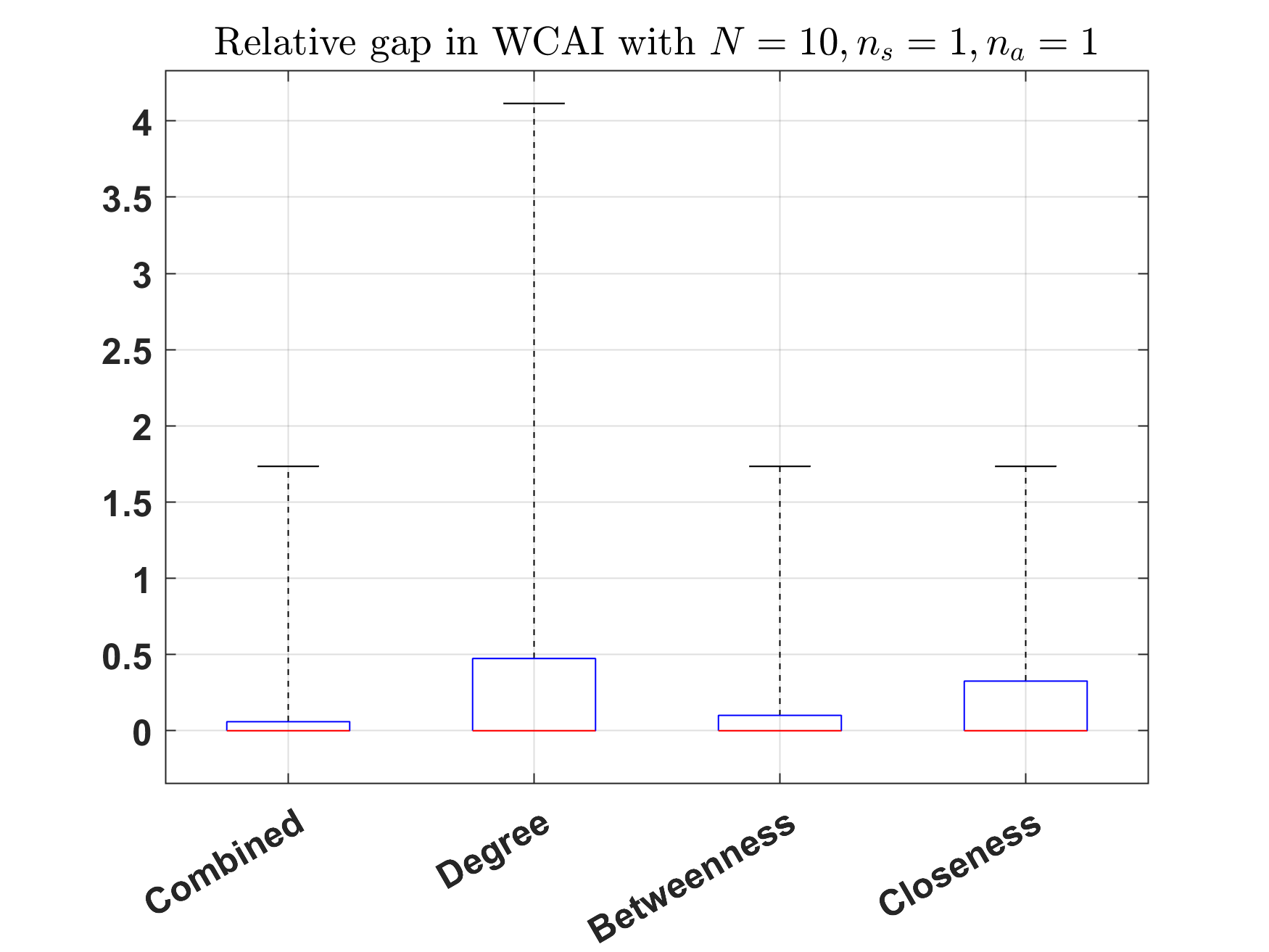}
    \end{subfigure}
    \begin{subfigure}{0.49\textwidth}
        \includegraphics[width=\textwidth]{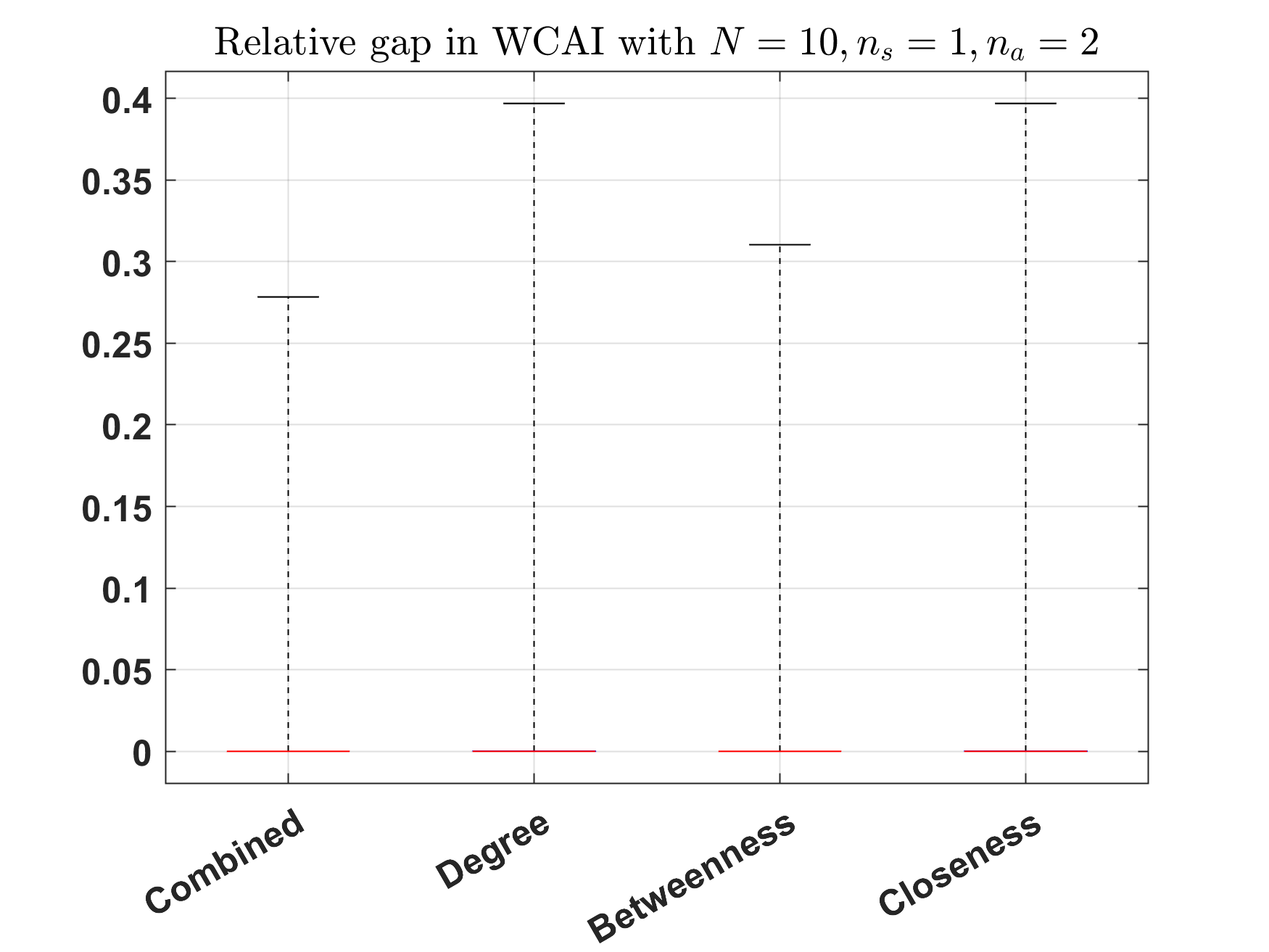}
    \end{subfigure}
    \begin{subfigure}{0.49\textwidth}
        \includegraphics[width=\textwidth]{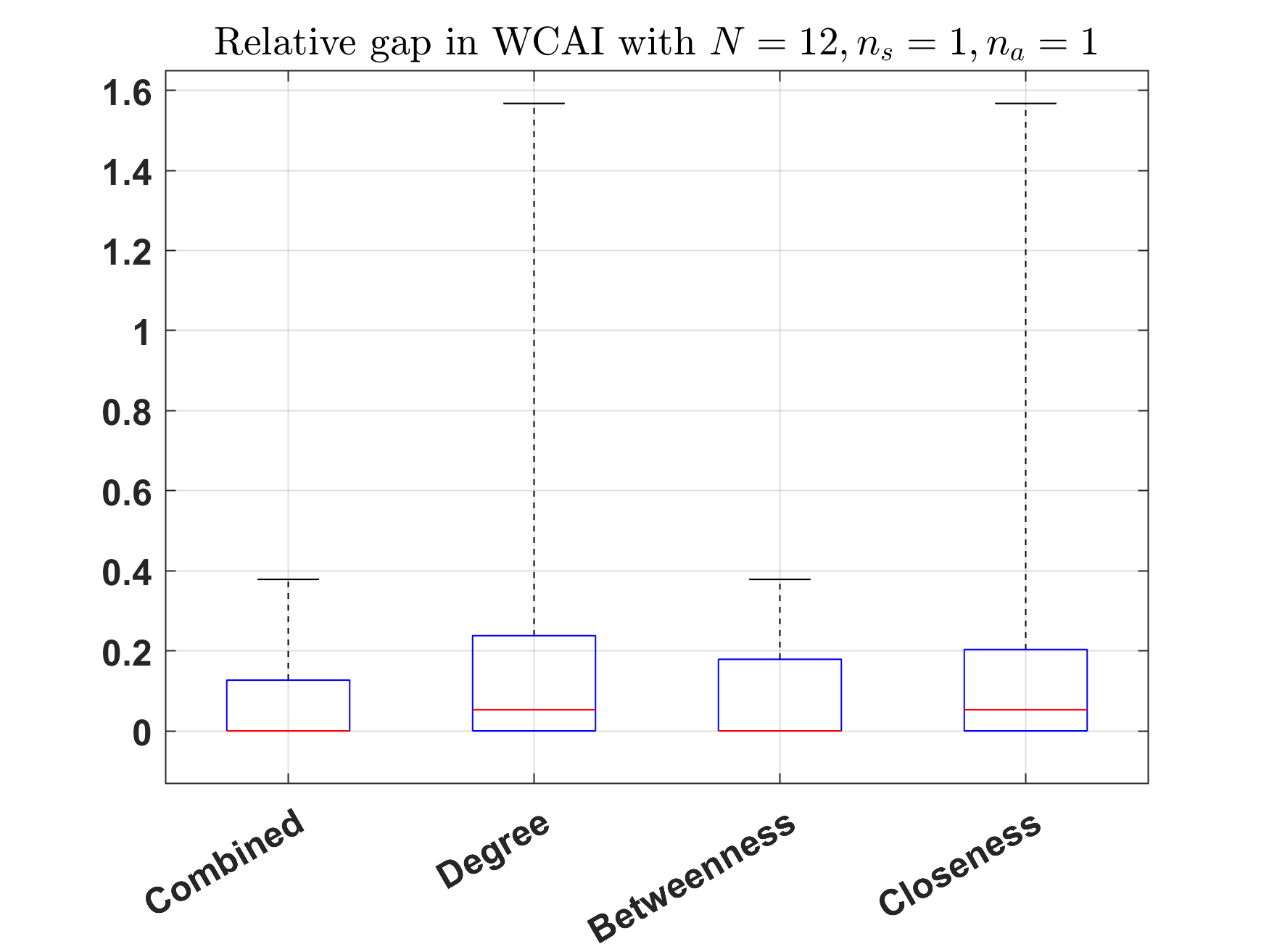}
    \end{subfigure}
    \begin{subfigure}{0.49\textwidth}
        \includegraphics[width=\textwidth]{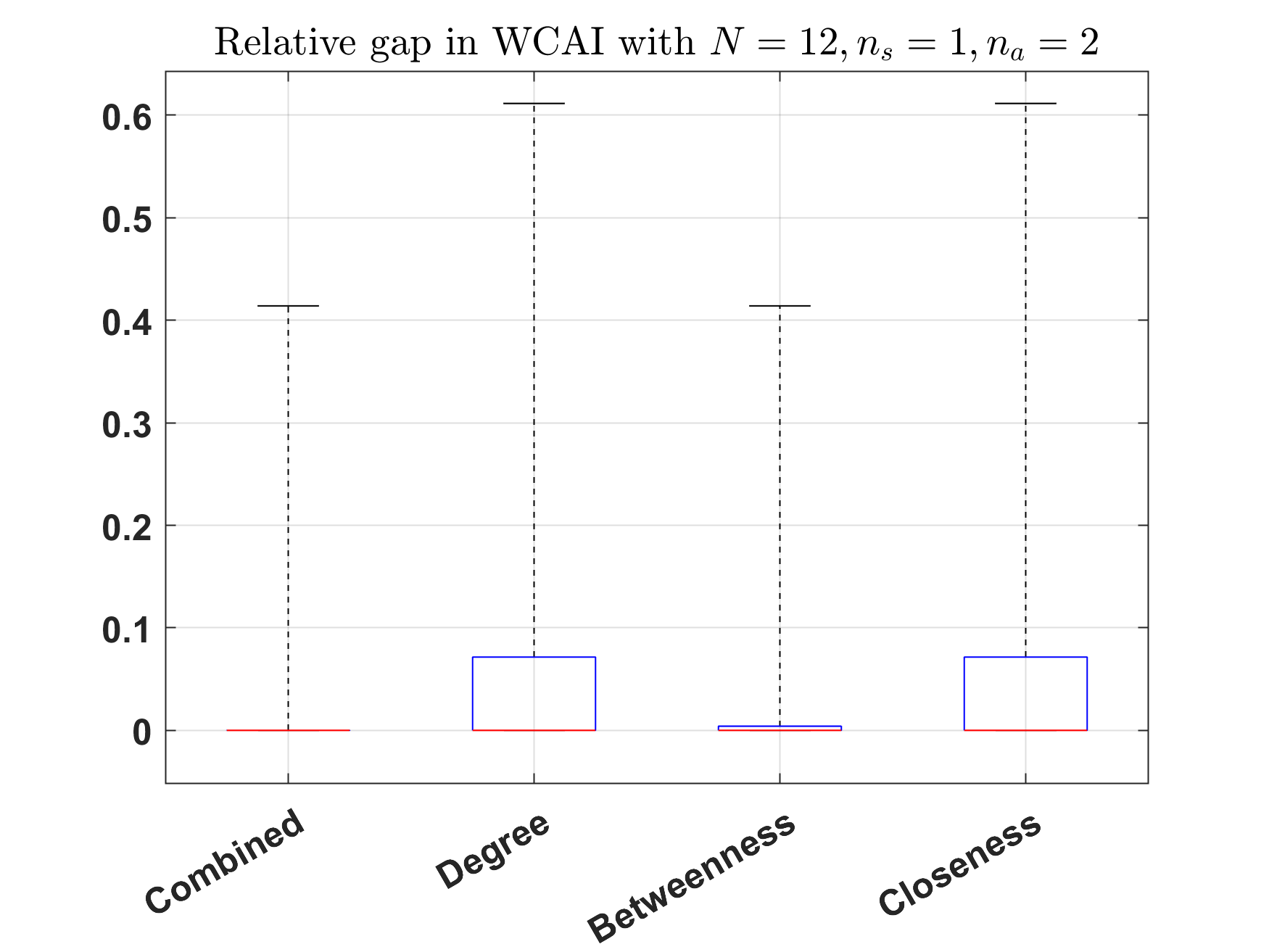}
    \end{subfigure}
    \begin{subfigure}{0.49\textwidth}
        \includegraphics[width=\textwidth]{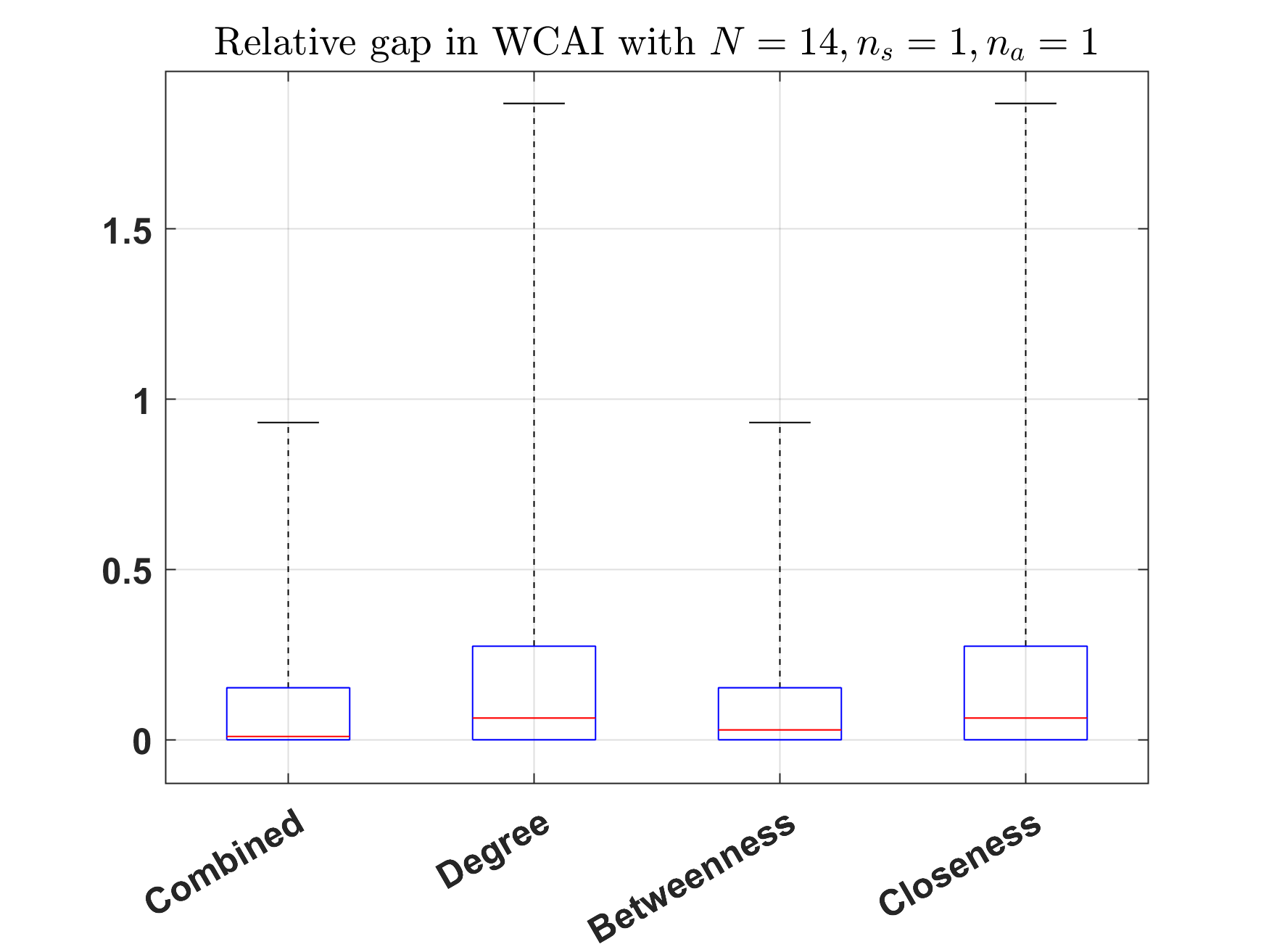}
    \end{subfigure}
    \begin{subfigure}{0.49\textwidth}
        \includegraphics[width=\textwidth]{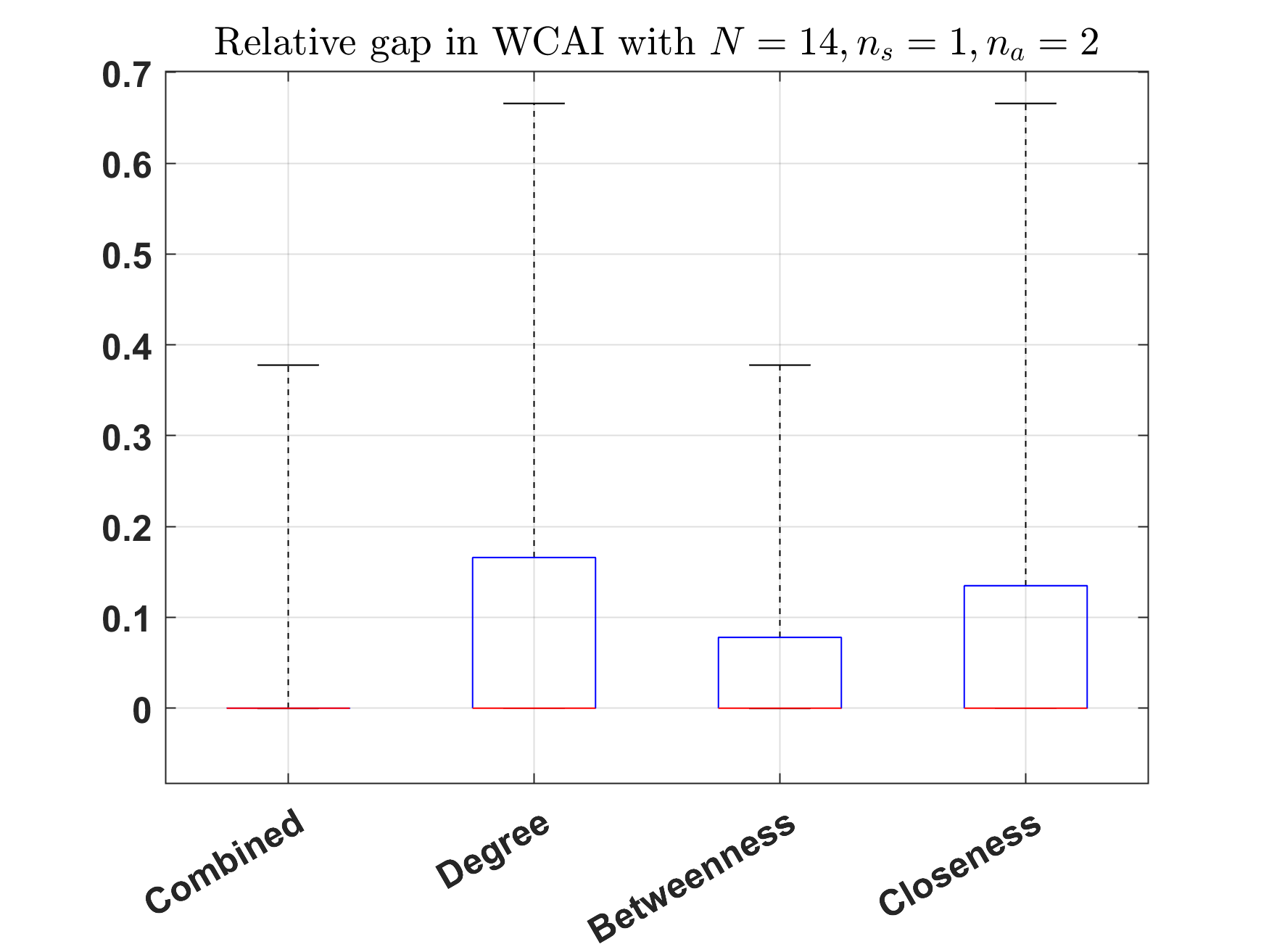}
    \end{subfigure}
    \caption{Relative gap of the worst-case attack impact between the optimal value and the value given by choosing monitor vertices based on centrality measures. The network size $N = \{ 10, \, 12, \, 14 \}$, the attack budget $n_a = \{1,\,2\}$, and the monitor budget $n_s = 1$. For each network size, $30$ Erdős–Rényi random graphs where an edge is included to connect two vertices with a probability of 0.5.}
    \label{fig:wcai10_14}
\end{figure}
\begin{figure}[!t]
    \centering
    \begin{subfigure}{0.49\textwidth}
        \includegraphics[width=\textwidth]{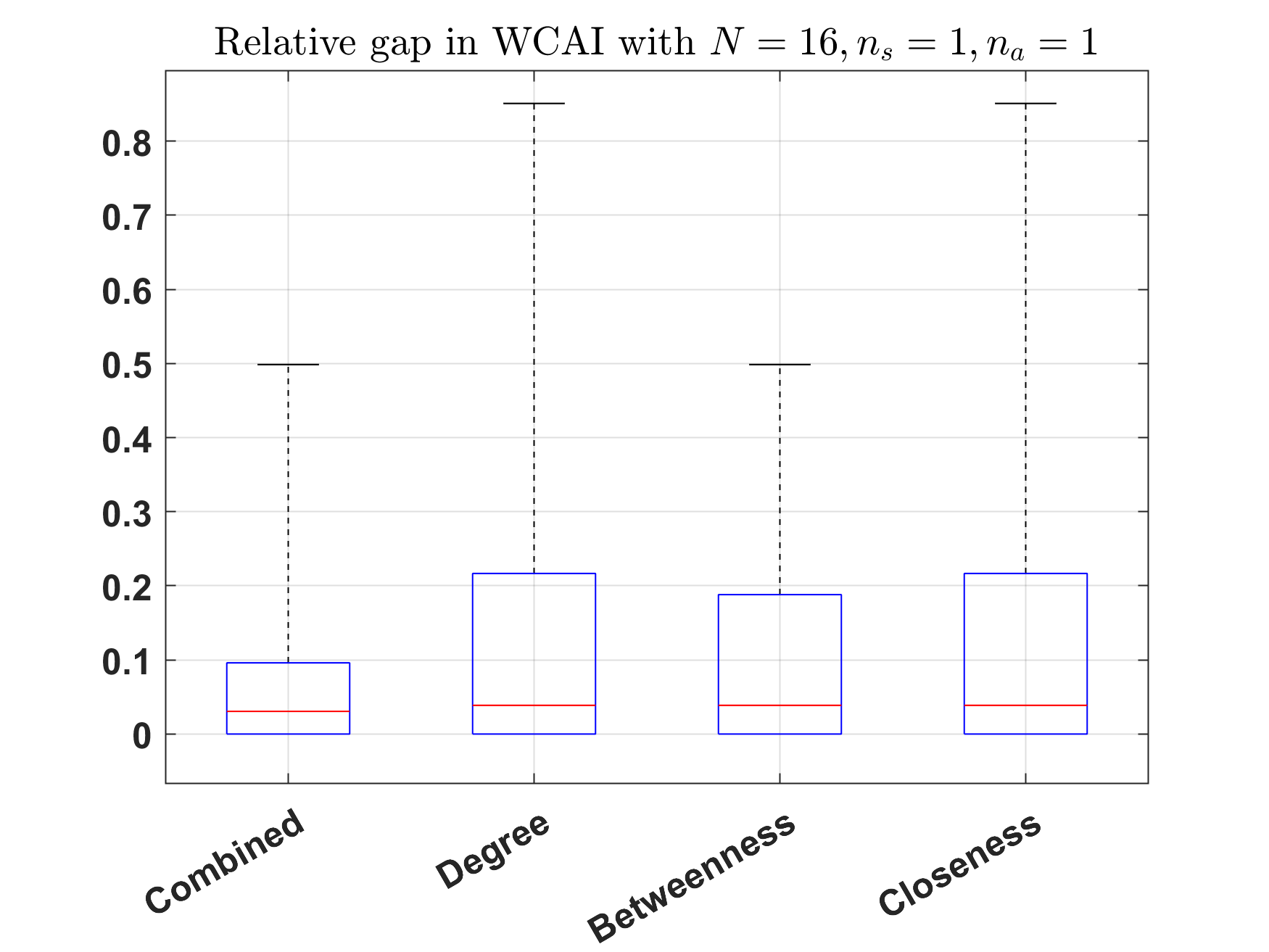}
    \end{subfigure}
    \begin{subfigure}{0.49\textwidth}
        \includegraphics[width=\textwidth]{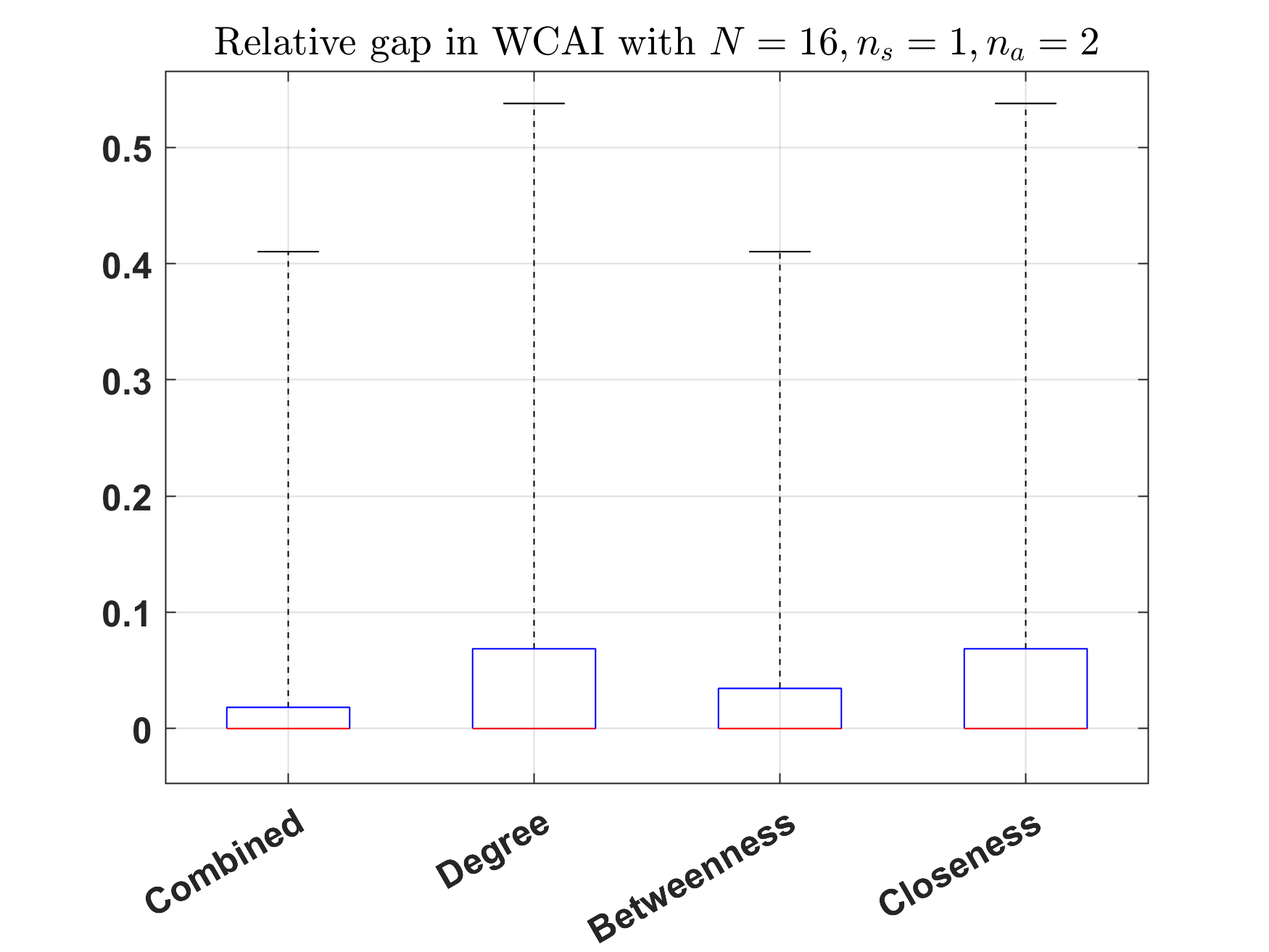}
    \end{subfigure}
    \begin{subfigure}{0.49\textwidth}
        \includegraphics[width=\textwidth]{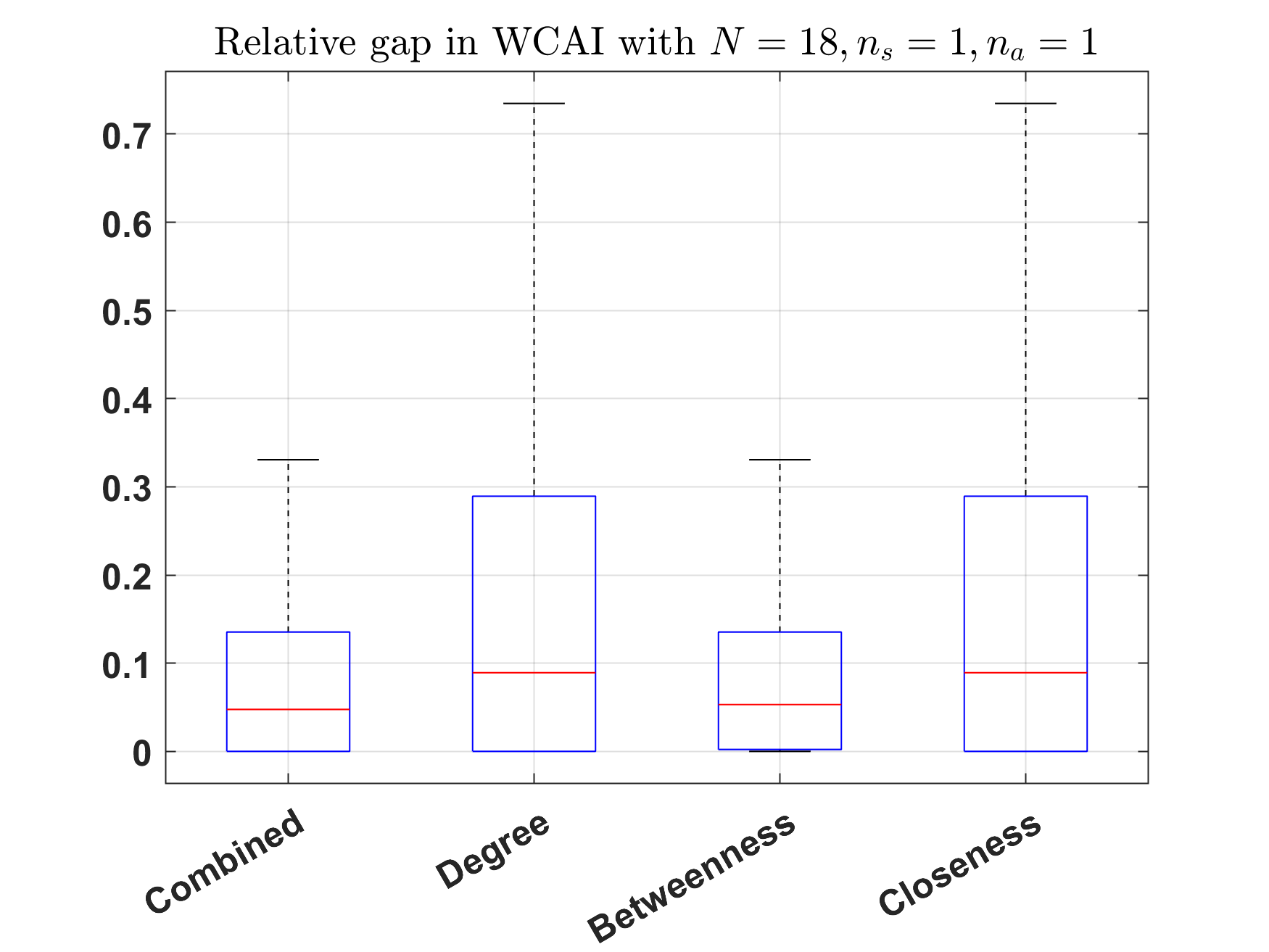}
    \end{subfigure}
    \begin{subfigure}{0.49\textwidth}
        \includegraphics[width=\textwidth]{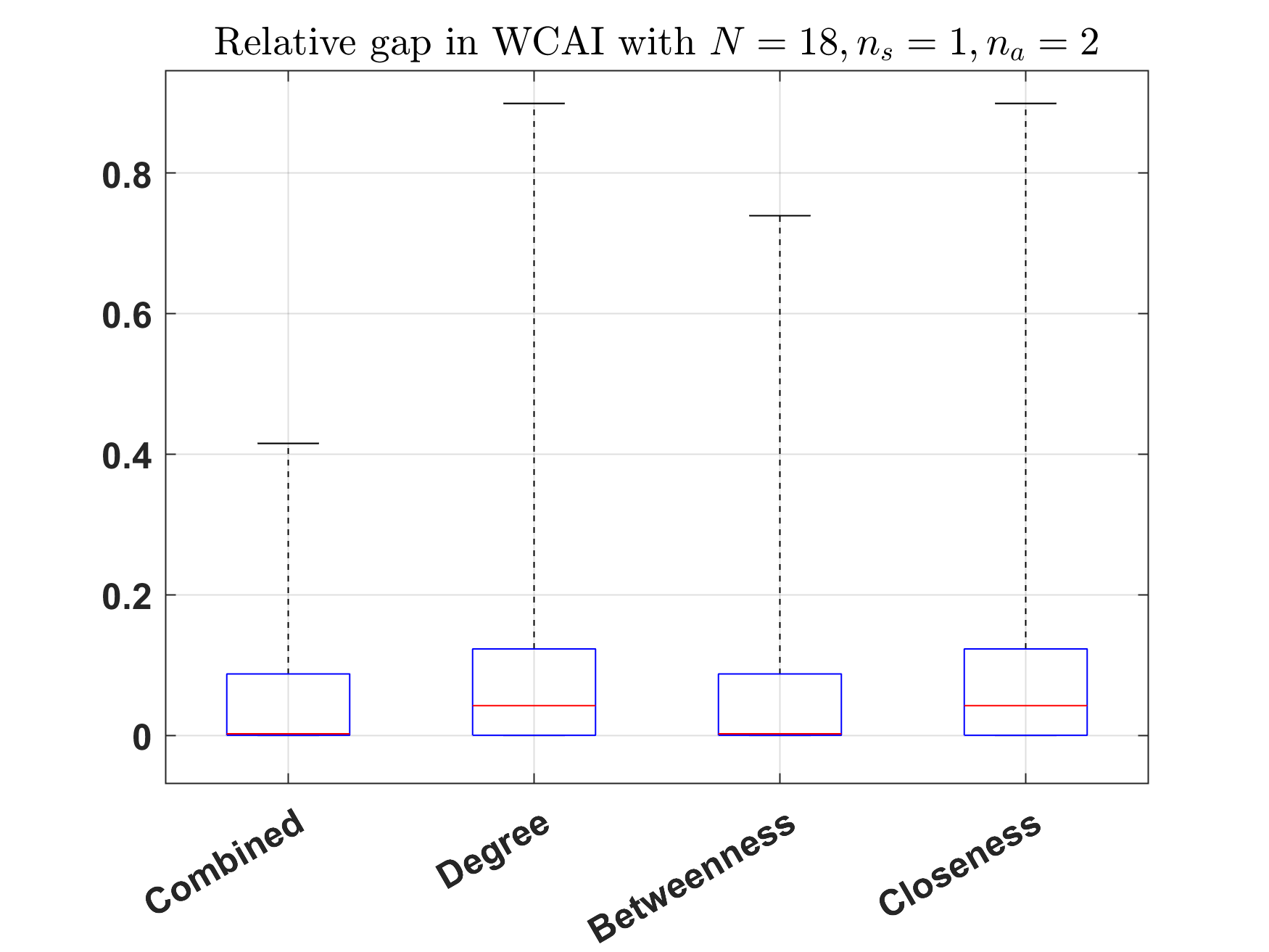}
    \end{subfigure}
    \begin{subfigure}{0.49\textwidth}
        \includegraphics[width=\textwidth]{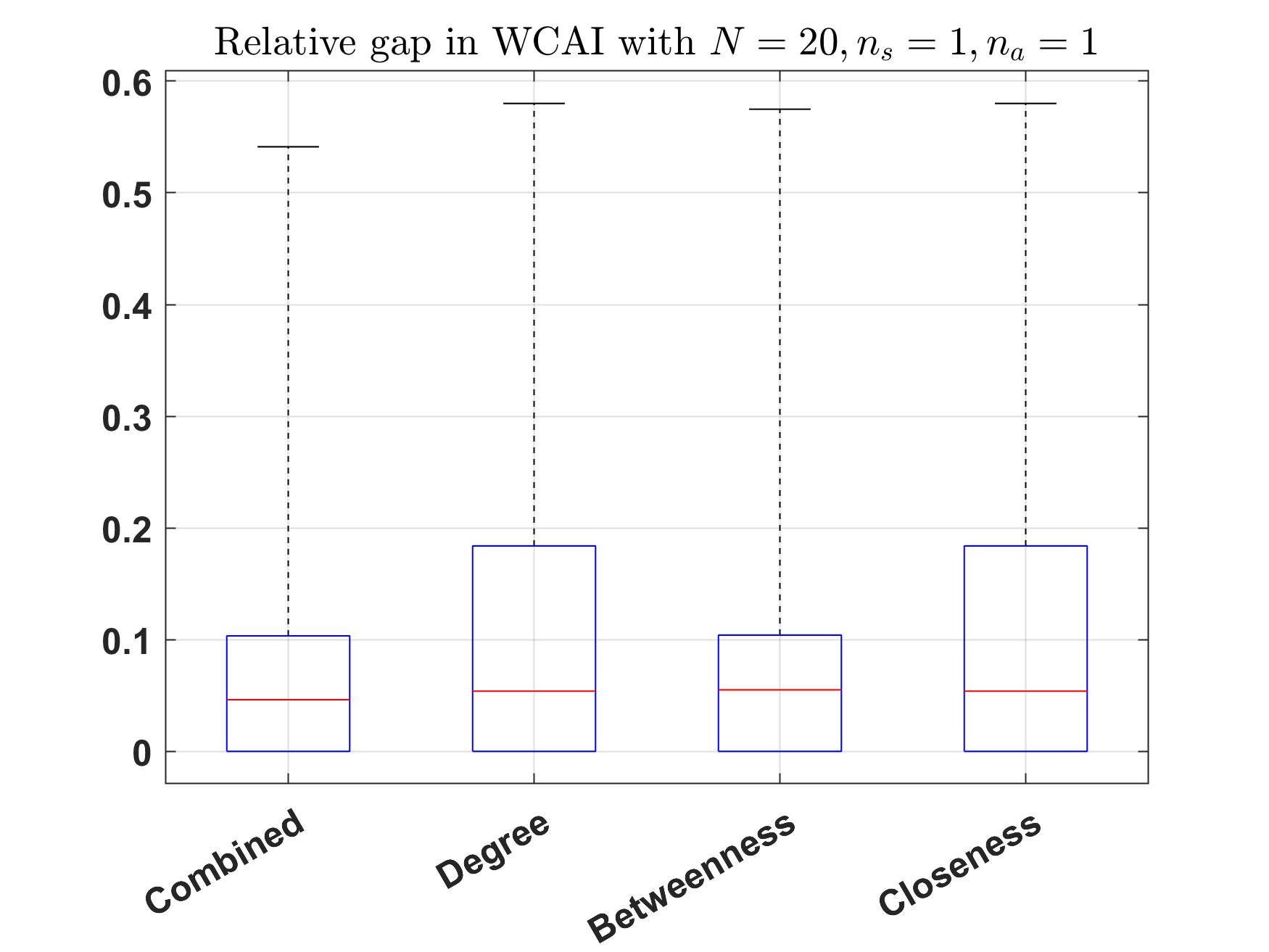}
    \end{subfigure}
    \begin{subfigure}{0.49\textwidth}
        \includegraphics[width=\textwidth]{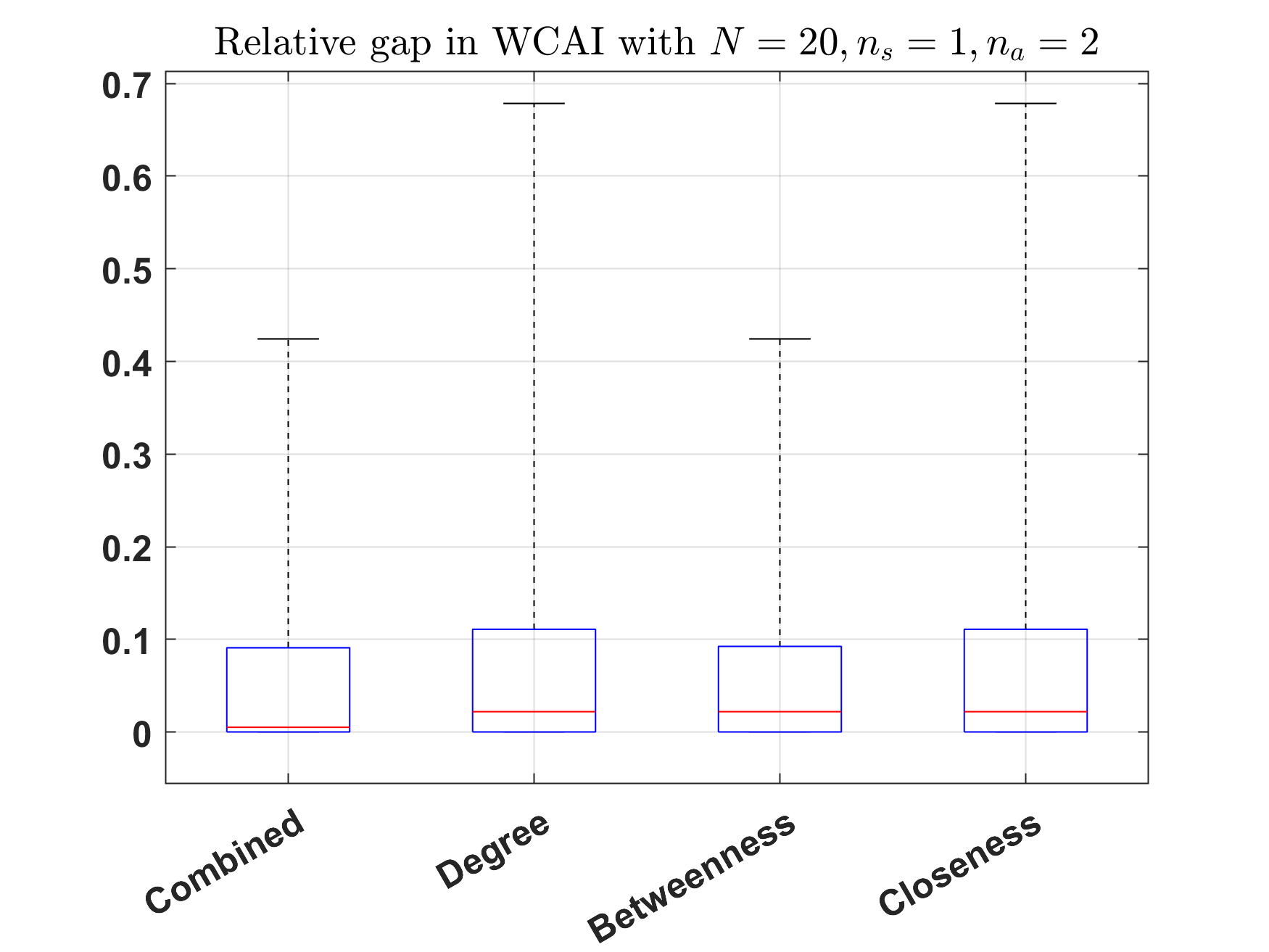}
    \end{subfigure}
    \caption{Relative gap of the worst-case attack impact between the optimal value and the value given by choosing monitor vertices based on centrality measures. The network size $N = \{ 16, \, 18, \, 20 \}$, the attack budget $n_a = \{1,\,2\}$, and the monitor budget $n_s = 1$. For each network size, $30$ Erdős–Rényi random graphs where an edge is included to connect two vertices with a probability of 0.5.}
    \label{fig:wcai16_20}
\end{figure}
\begin{figure}[!t]
    \centering
    \begin{subfigure}{0.49\textwidth}
        \includegraphics[width=\textwidth]{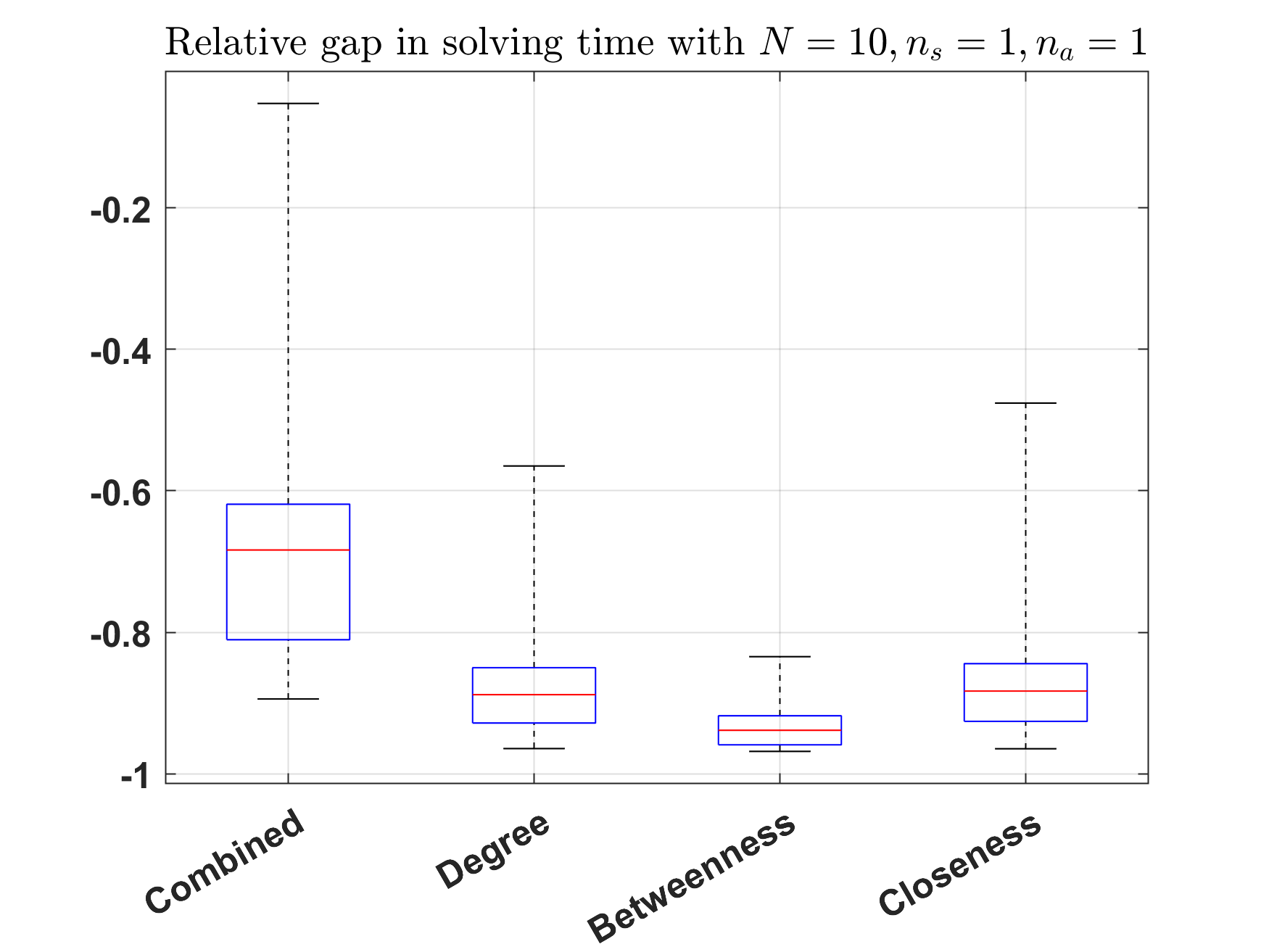}
    \end{subfigure}
    \begin{subfigure}{0.49\textwidth}
        \includegraphics[width=\textwidth]{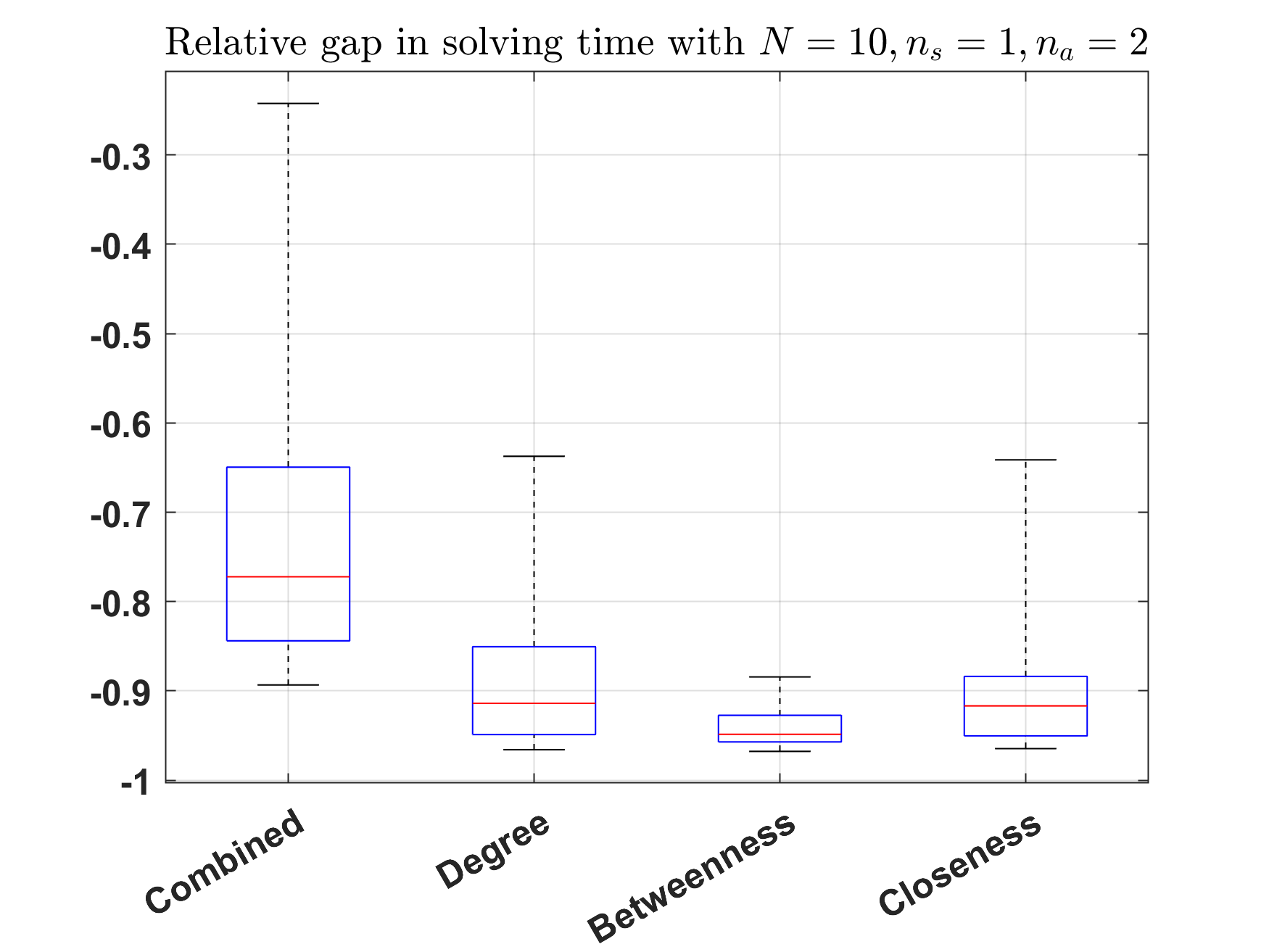}
    \end{subfigure}
    \begin{subfigure}{0.49\textwidth}
        \includegraphics[width=\textwidth]{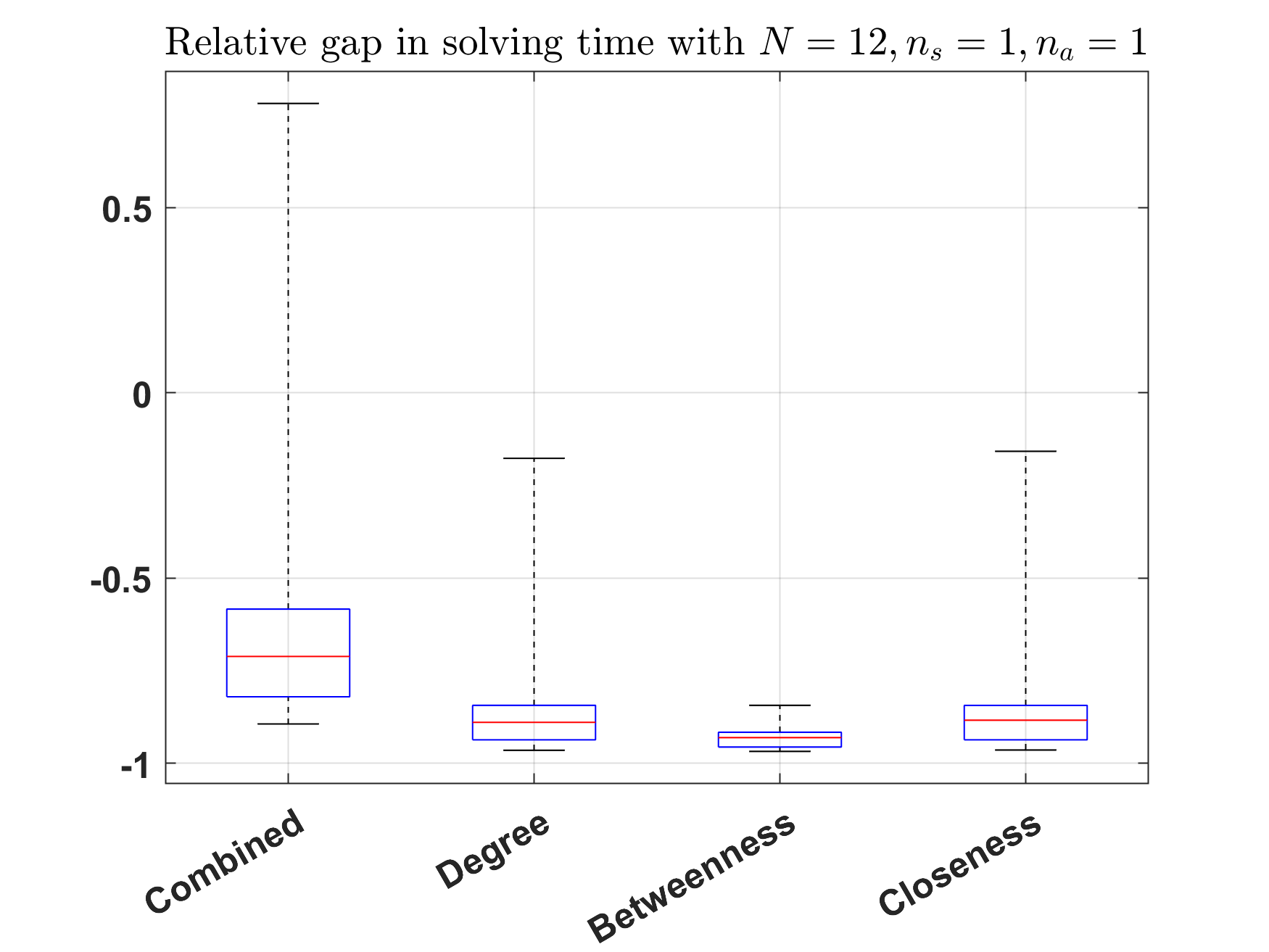}
    \end{subfigure}
    \begin{subfigure}{0.49\textwidth}
        \includegraphics[width=\textwidth]{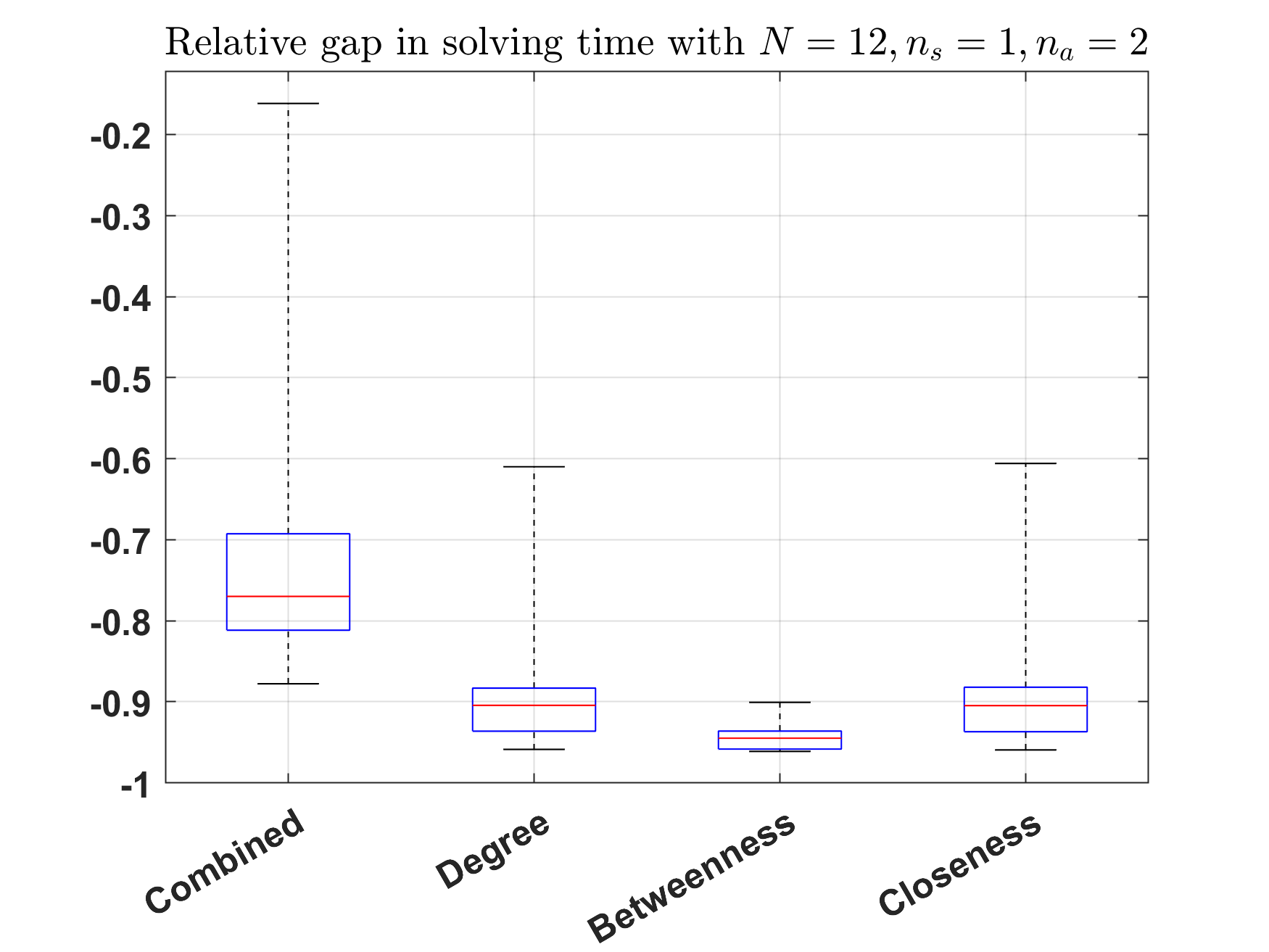}
    \end{subfigure}
    \begin{subfigure}{0.49\textwidth}
        \includegraphics[width=\textwidth]{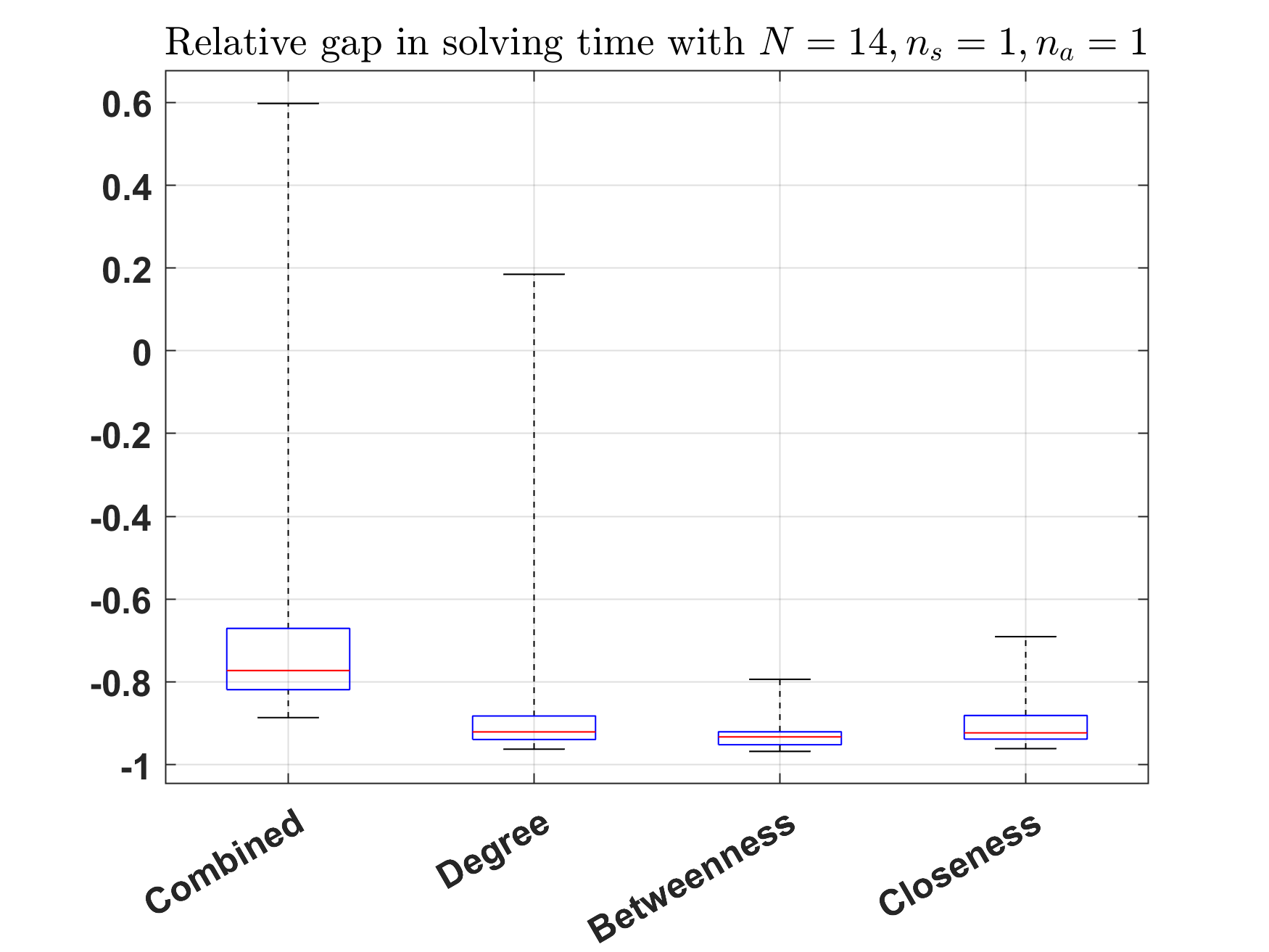}
    \end{subfigure}
    \begin{subfigure}{0.49\textwidth}
        \includegraphics[width=\textwidth]{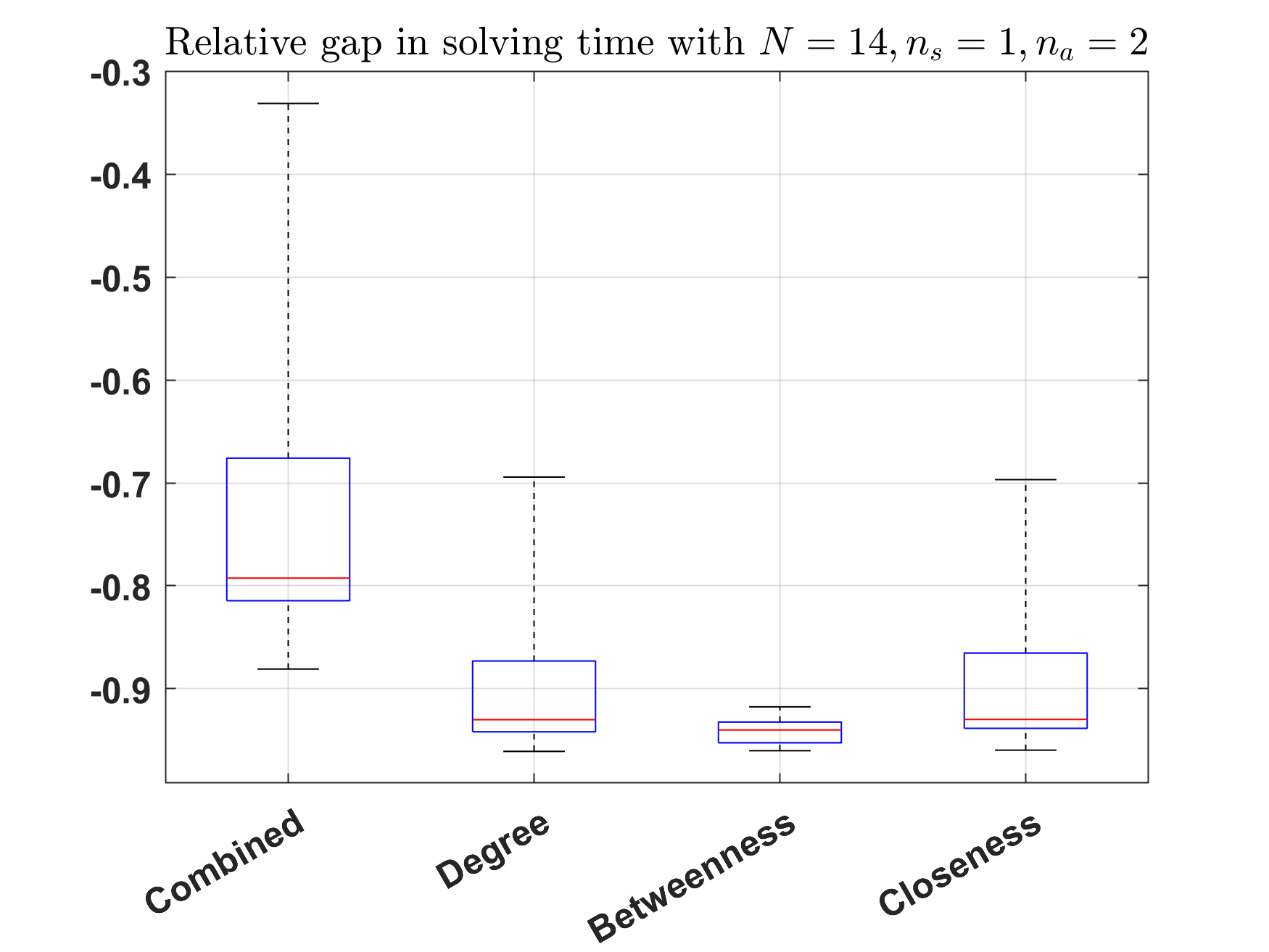}
    \end{subfigure}
    \caption{Relative gap of the solving time between the optimal value and the value given by choosing monitor vertices based on centrality measures. The network size $N = \{ 10, \, 12, \, 14 \}$, the attack budget $n_a = \{1,\,2\}$, and the monitor budget $n_s = 1$. For each network size, $30$ Erdős–Rényi random graphs where an edge is included to connect two vertices with a probability of 0.5.}
    \label{fig:time10_14}
\end{figure}
\begin{figure}[!t]
    \centering
    \begin{subfigure}{0.49\textwidth}
        \includegraphics[width=\textwidth]{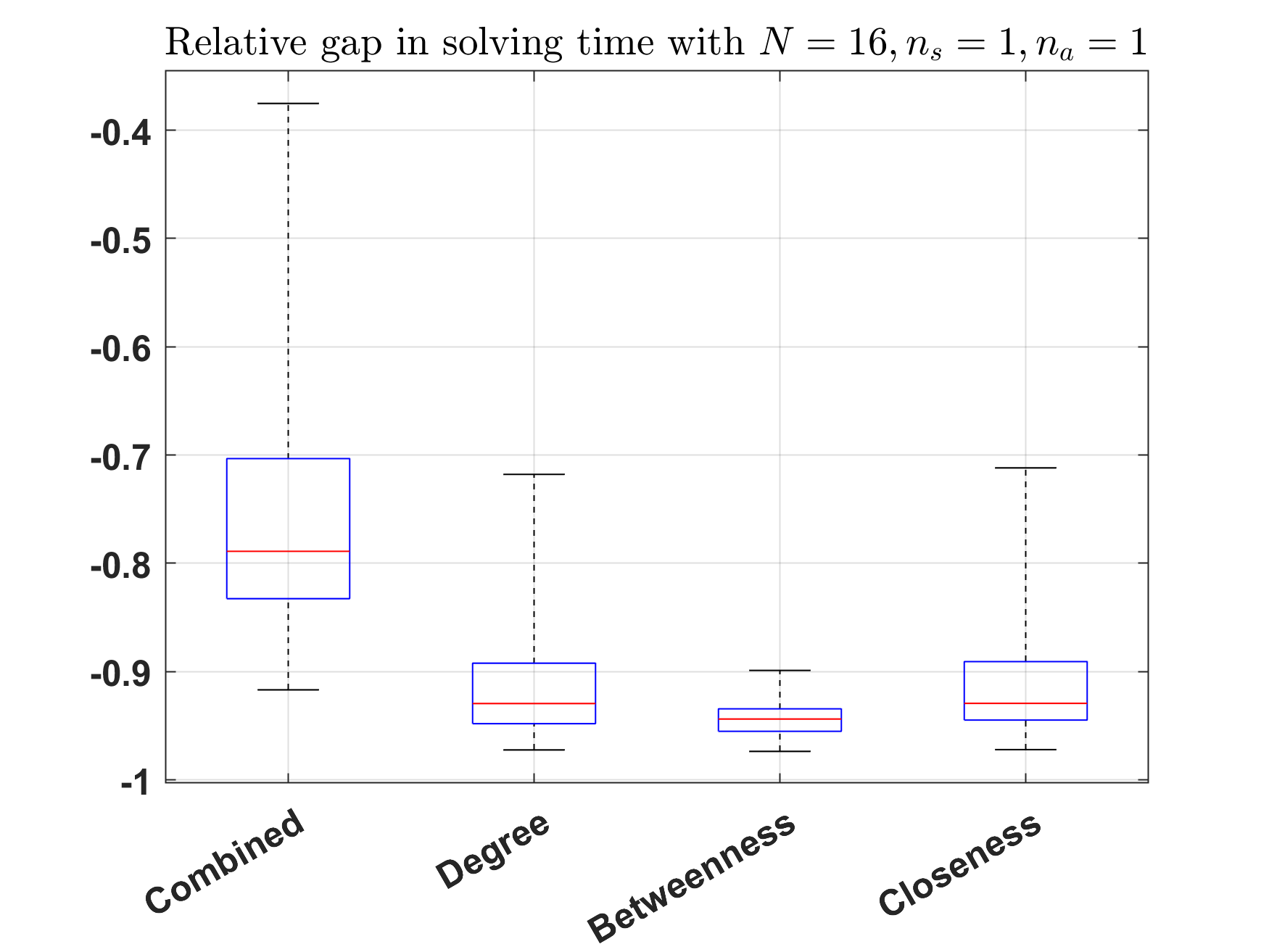}
    \end{subfigure}
    \begin{subfigure}{0.49\textwidth}
        \includegraphics[width=\textwidth]{figs/16node30graph1monitor1attacktime.png}
    \end{subfigure}
    \begin{subfigure}{0.49\textwidth}
        \includegraphics[width=\textwidth]{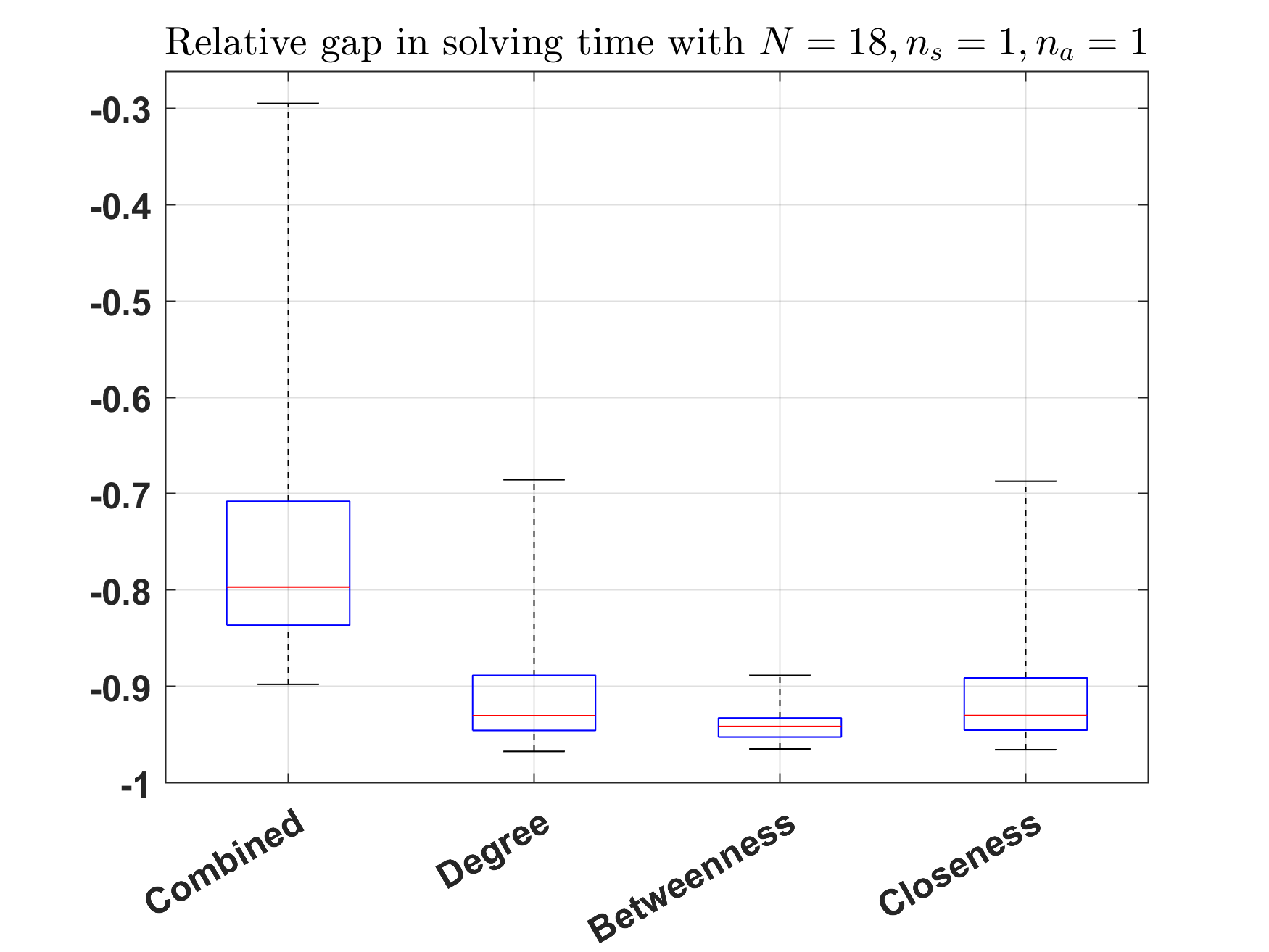}
    \end{subfigure}
    \begin{subfigure}{0.49\textwidth}
        \includegraphics[width=\textwidth]{figs/18node30graph1monitor1attacktime.png}
    \end{subfigure}
    \begin{subfigure}{0.49\textwidth}
        \includegraphics[width=\textwidth]{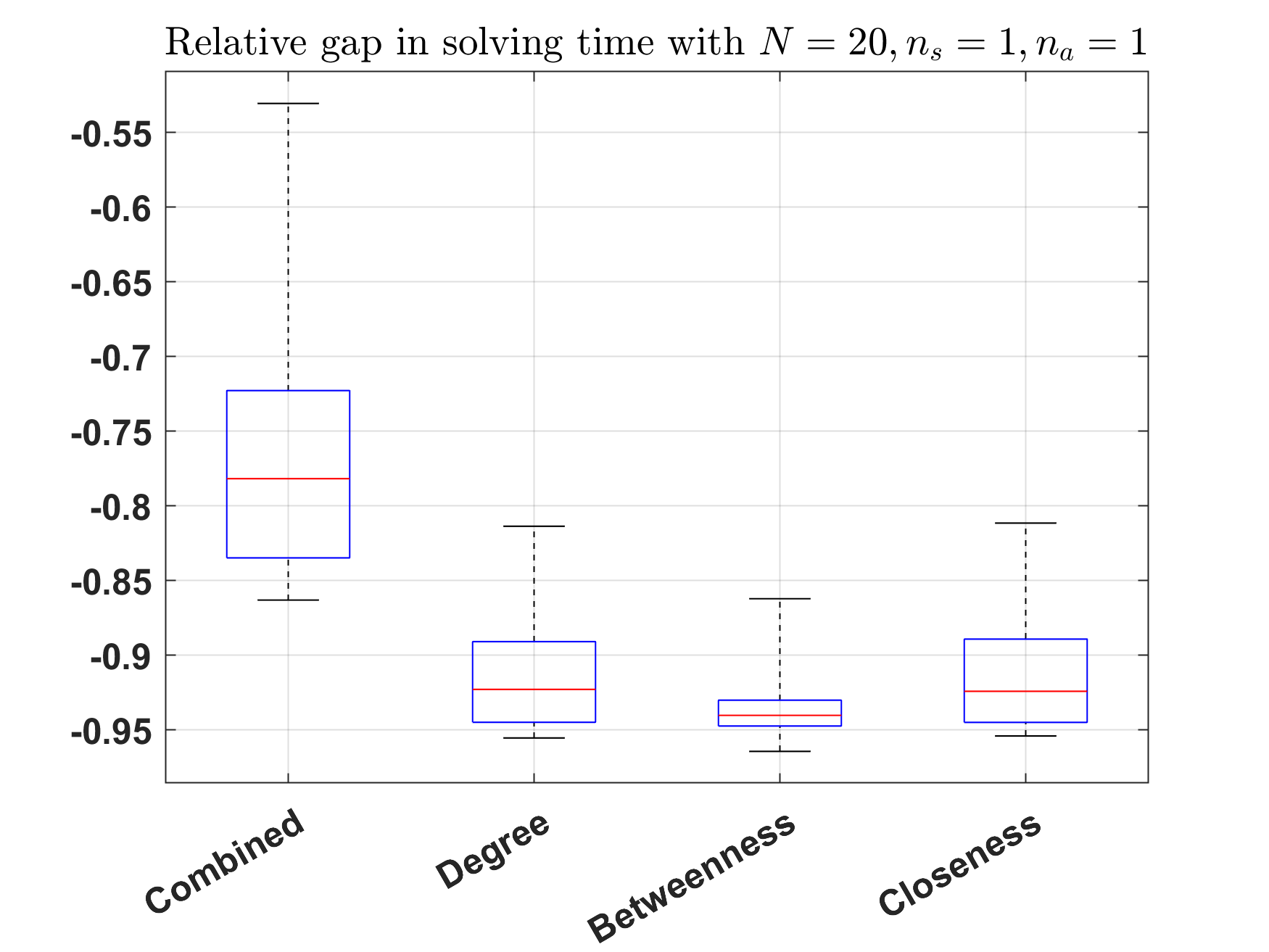}
    \end{subfigure}
    \begin{subfigure}{0.49\textwidth}
        \includegraphics[width=\textwidth]{figs/20node30graph1monitor1attacktime.png}
    \end{subfigure}
    \caption{Relative gap of the solving time between the optimal value and the value given by choosing monitor vertices based on centrality measures. The network size $N = \{ 16, \, 18, \, 20 \}$, the attack budget $n_a = \{1,\,2\}$, and the monitor budget $n_s = 1$. For each network size, $30$ Erdős–Rényi random graphs where an edge is included to connect two vertices with a probability of 0.5.}
    \label{fig:time16_20}
\end{figure}

~\newpage ~\newpage ~\newpage 
\subsection{IEEE 14-bus system}
In this part, we did an experiment on an IEEE benchmark of power systems which is the IEEE 14-bus system, depicted in Figure~\ref{fig:ieee14}. The system includes 14 buses and 20 transmission lines.
The behavior of a bus $i \in \{1,2,\ldots,14\}$ can be described by the so-called swing equation \cite{tegling2018fundamental}:
\begin{align}
    I_i \ddot \theta_i(t) + D_i \dot \theta_i(t) = -\sum_{j \in \Nc_i} P_{ij}(t), \label{bus_dyn}
\end{align}
where $I_i$ and $D_i$ are the inertia and damping coefficients, respectively, and $P_{ij}$ is the active power flow from bus $j$ to bus $i$. 
Considering that there are no power losses and $V_i = |V_i|e^{j p_i}~(j^2 = -1)$ and $\theta_i$ be the complex voltage and the phase angle of the bus $i$, respectively. 
The active power flow $P_{ij}(t)$ from bus $j$ to bus $i$ is given by 
\begin{align}
    P_{ij}(t) = -\ell_{ij} \sin(\theta_i - \theta_j), \label{power_ij}
\end{align}
where $-\ell_{ij} \in \Rbb_+$ is the susceptance of the power transmission line connecting bus $i$ with bus $j$.
Those parameters consisting of line susceptance $-\ell_{ij}$, inertia $I_i$, and damping $D_i$ can be found at \cite{ieee14bus} and are listed in Table~\ref{tab:ieee_parameters}.
Since the phase angles usually are close, we can linearize \eqref{power_ij} and rewrite the dynamics \eqref{bus_dyn} of bus $i$ as follows
\begin{align}
    I_i \ddot \theta_i(t) + D_i \dot \theta_i(t) =  \sum_{j \in \Nc_i} \ell_{ij} \Big( \theta_i(t) - \theta_j(t) \Big) ,
\end{align}
which can be rewritten as follows:
\begin{align}
    \frac{\text{d}}{\text{d} \, t}
    \ba{c}
    \dot \theta(t) \\ \theta(t)
    \ea = \ba{cc}
    0 & I \\
    -I^{-1}\,L & -I^{-1}\, D \ea 
    \ba{c}
    \dot \theta(t) \\ \theta(t)
    \ea, 
\end{align}
where $\dot \theta(t) = [\dot \theta_1(t),\dot \theta_2(t),\ldots,\dot \theta_N(t)]^\top$, $\theta(t) = [\theta_1(t),\theta_2(t),\ldots,\theta_N(t)]^\top$. Further, $L$ is a Laplacian matrix representing the network, $I= \text{diag}(I_i)$, and $D = \text{diag}(D_i)$.

By applying the computation of centrality measures in the previous section, we found that the bus number $4$ always has the highest value. However, the optimal value computed by \eqref{mixedintegersdpoptimal} gives us the optimal bus number $2$. Next, we compare the WCAI of all the monitor buses depicted in Figure~\ref{fig:ieee14_result}. We observe that the relative gap in WCAI of the bus number $4$ compared to the optimal bus $2$ is under 2\%, an acceptable value. Meanwhile, the other buses provide much higher relative gaps, for example, bus $8$ gives us 10\% relative gap.

\begin{figure}[!t]
    \centering
    \includegraphics[width=0.8\textwidth]{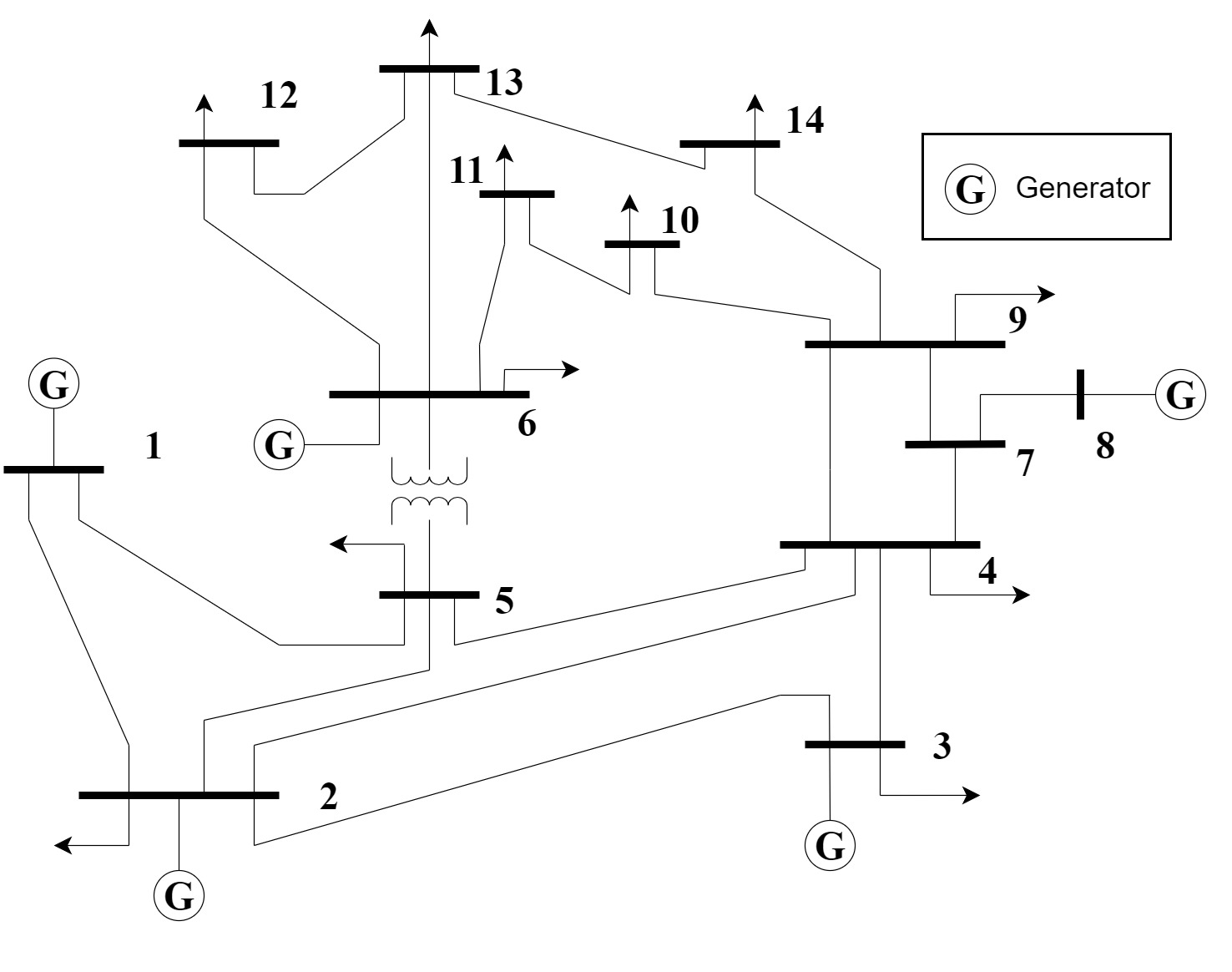}
    \caption{IEEE 14-bus system.}
    \label{fig:ieee14}
\end{figure}
\begin{table}[!t]
    \centering
    \begin{tabular}{|c|c|c|}
        \hline
        \hspace{0.7cm}Susceptance\hspace{0.7cm} & \hspace{0.7cm}Inertial\hspace{0.7cm} & \hspace{0.7cm}Damping\hspace{0.7cm} \\ \hline 
        $-\ell_{12} = 8.2838$ & $I_1 = 1.060$ & $D_1 = 0$ \\  \hline
        $-\ell_{15}=31.2256$  & $I_2 = 1.045$ & $D_2 = 4.98$ \\  \hline
        $-\ell_{23}=27.7158$ & $I_3 = 1.010$ & $D_3 = 12.72$\\  \hline
        $-\ell_{24}=24.6848$ & $I_4 = 1.019$ & $D_4 = 10.33$ \\  \hline
        $-\ell_{25}=24.3432$ & $I_5 = 1.020$ & $D_5 = 8.78$\\  \hline
        $-\ell_{34}=23.9442$ & $I_6 = 1.070$ & $D_6 = 14.22$ \\  \hline
        $-\ell_{45}=5.8954$  & $I_7 = 1.062$ & $D_7 = 13.37$ \\  \hline
        $-\ell_{47}=29.2768$ & $I_8 = 1.090$ & $D_8 = 13.36$ \\  \hline
        $-\ell_{49}=77.8652$ & $I_9 = 1.056$  &  $D_9 = 14.94$ \\  \hline
        $-\ell_{56}=35.2828$ & $I_{10} = 1.051$ &  $D_{10} = 15.10$\\  \hline
        $-\ell_{6,11}=35.2828$ & $I_{11} = 1.057$ & $D_{11} = 14.79$ \\  \hline
        $-\ell_{6,12}=35.8134$ & $I_{12} = 1.055$ & $D_{12} = 15.07$ \\  \hline
        $-\ell_{6,13}=18.2378$ & $I_{13} = 1.050$ & $D_{13} = 15.16$ \\  \hline
        $-\ell_{78}=24.6610$ & $I_{14} = 1.036$ & $D_{14} = 16.04$ \\  \hline
        $-\ell_{79}=15.4014$ & & \\  \hline
        $-\ell_{9,10}=11.8300$ & & \\  \hline
        $-\ell_{9,14}=37.8532$ & & \\  \hline
        $-\ell_{10,11}=26.8898$ & & \\  \hline
        $-\ell_{12,13}=27.9832$ & & \\  \hline
        $-\ell_{13,14}=48.7228$ & & \\  \hline
    \end{tabular}
    \caption{Parameters of the IEEE 14-bus system.}
    \label{tab:ieee_parameters}
\end{table}

\begin{figure}[!t]
    \centering
    \begin{subfigure}
        {0.49\textwidth}
        \includegraphics[width=\textwidth]{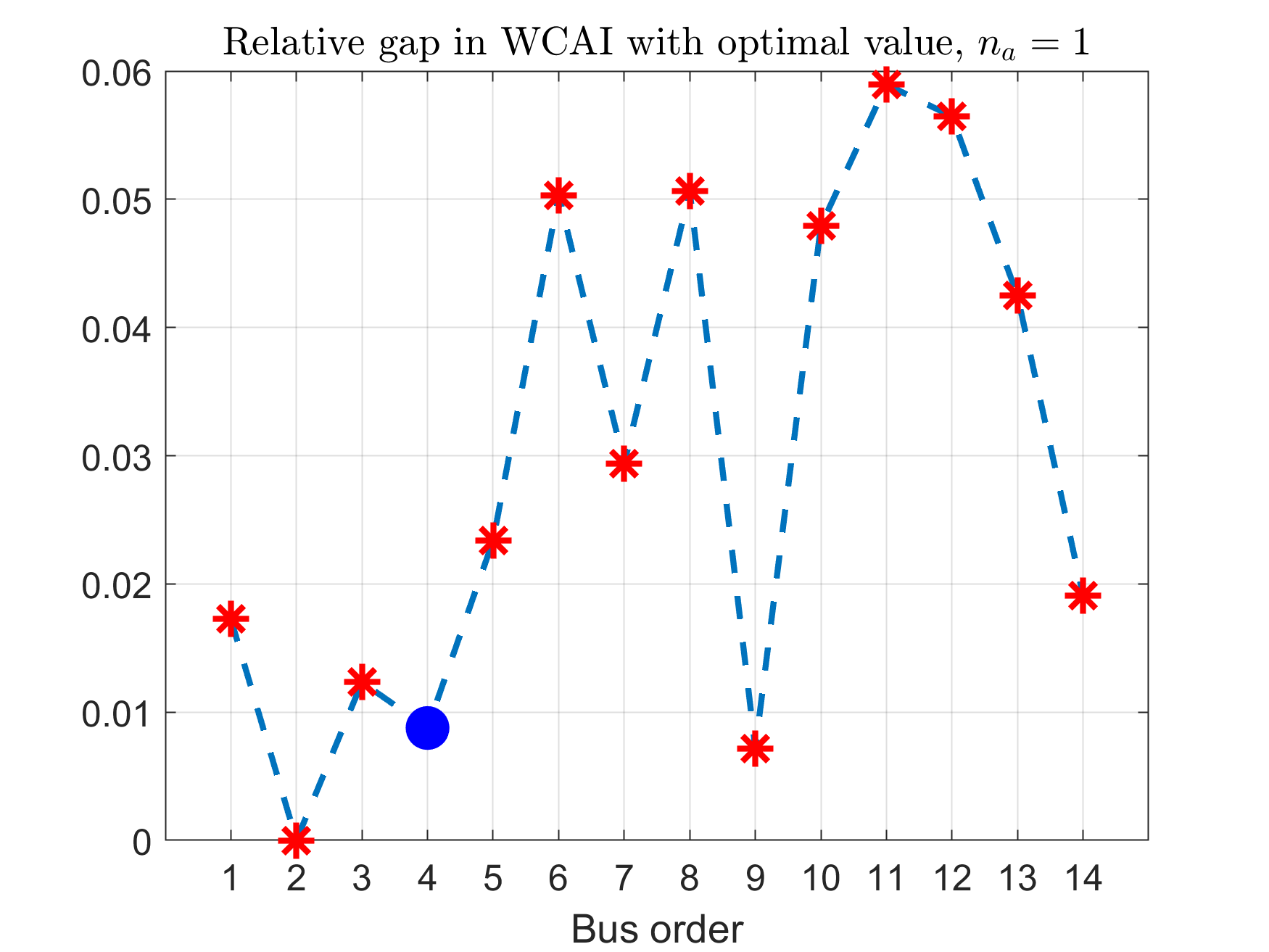}
    \end{subfigure}
    \begin{subfigure}
        {0.49\textwidth}
        \includegraphics[width=\textwidth]{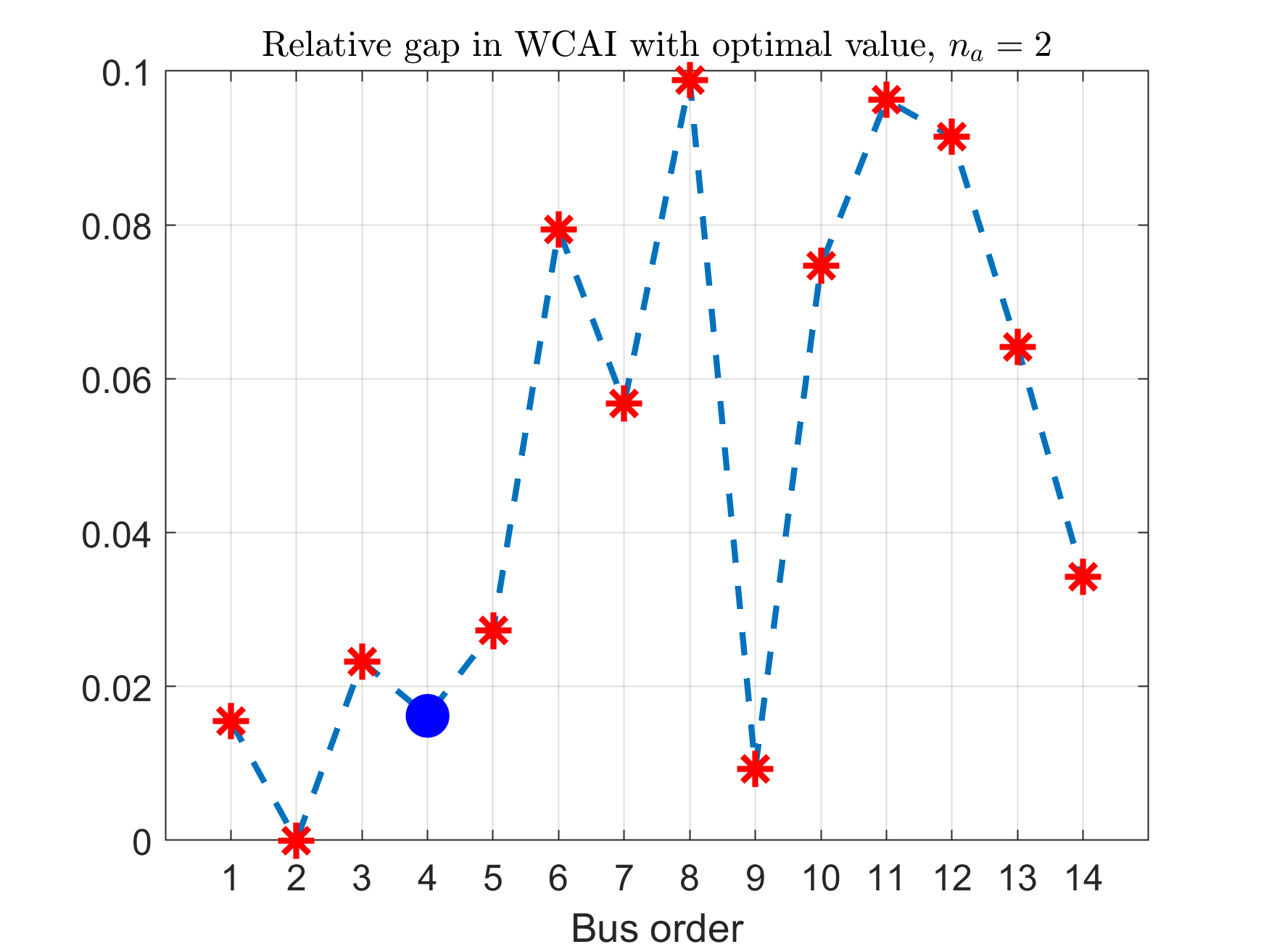}
    \end{subfigure}
    \caption{Relative gap of the worst-case attack impact between the optimal value and the value given by choosing a monitor vertex based on centrality measures (blue dot at bus 4).}
    \label{fig:ieee14_result}
\end{figure}

~\newpage
\section{Conclusion}
\label{sec:conclusion}
In conclusion, we explored the security allocation problem within a networked control system featuring two opposing agents: a defender and a malicious adversary. The adversary aims to maximize the WCAI on the network performance while remaining undetected, by launching stealthy data injection attacks. Conversely, the defender’s objective is to allocate security resources to minimize the WCAI. The proposed approach was centrality measures-based security allocation using degree, betweenness, and closeness centrality, in addition to a combined centrality measure approach using all three. The conclusion one can draw from this study is that the combined approach, although it could only perform as well as the best monitor set derived using the three centrality measures for each graph, outperformed them on average across the 30 Erdős–Rényi random graphs (a probability of 0.5) with 10 to 20 vertices considered in this paper. This approach can be utilized for security allocation in networks if the relative gap of $10\%$ in terms of WCAI to the optimal solution is acceptable. We can also conclude that because the solver times are significantly shorter for the combined centrality measure approach than computing the optimal solution, this method is advantageous when hardware capabilities are limited. 
Additionally, among the three centrality measures investigated, the betweenness centrality performed the best. Similar to the combined centrality measure approach, we can conclude that betweenness centrality can be utilized for the security allocation in networks derived using Erdős–Rényi random graphs (a probability of 0.5) with 10 to 20 vertices. Moreover, the solver time for only using betweenness centrality is a third or shorter than the combined approach.

\begin{credits}
\subsubsection{\ackname} 
This work is supported by the Swedish Research Council under the
grant 2021-06316 and by the Swedish Foundation for Strategic Research.
\end{credits}

\bibliographystyle{splncs04}
\bibliography{mybib}

\end{document}

%% file: deflatex3.tex
\def\Ac{{\mathcal A}}

\def\Ec{{\mathcal E}}

\def\Gc{{\mathcal G}}

\def\Lc{{\mathcal L}}

\def\Mc{{\mathcal M}}

\def\Nc{{\mathcal N}}

\def\Rbb{{\mathbb R}}

\def\Sbb{{\mathbb S}}

\def\Vc{{\mathcal V}}

\def\0{{\bf 0}}

\newcommand{\bitem}{\begin{itemize}}
\newcommand{\eitem}{\end{itemize}}
\newcommand{\btabular}{\begin{tabular}}
\newcommand{\etabular}{\end{tabular}}
\newcommand{\bcenter}{\begin{center}}
\newcommand{\ecenter}{\end{center}}
\newcommand{\bea}{\begin{eqnarray}}
\newcommand{\eea}{\end{eqnarray}}
\newcommand{\bean}{\begin{eqnarray*}}
\newcommand{\eean}{\end{eqnarray*}}
\newcommand{\ba}{\left[ \begin{array}}
\newcommand{\ea}{\\ \end{array} \right]}
\newcommand{\bear}{\begin{array}}
\newcommand{\eear}{\\ \end{array}}

\newcommand{\non}{\nonumber}

\newcommand*{\QEDB}{\hfill\ensuremath{\blacksquare}}%
\newcommand*{\QET}{\hfill\ensuremath{\triangleleft}}
\newcommand{\norm}[1]{\left\lVert#1\right\rVert}

\newcounter{subequation}
\def\beasub{\addtocounter{equation}{+1}
\setcounter{subequation}{\value{equation}}
\setcounter{equation}{0}
\renewcommand{\theequation}{\arabic{subequation}\alph{equation}}
\begin{eqnarray}}
\def\eeasub{\end{eqnarray}
\setcounter{equation}{\value{subequation}}
\renewcommand{\theequation}{\arabic{equation}}}

\def\inf{\operatornamewithlimits{inf\vphantom{p}}}


%% file: mybib.bib
@inproceedings{lofberg2004yalmip,
  title={YALMIP: A toolbox for modeling and optimization in MATLAB},
  author={Lofberg, Johan},
  booktitle={2004 IEEE International Conference on Robotics and Automation (IEEE Cat. No. 04CH37508)},
  pages={284--289},
  year={2004},
  organization={IEEE},
  doi={10.1109/CACSD.2004.1393890}
}

@article{purposeadver,
  author={Umsonst, David and Sarıtaş, Serkan and Dán, György and Sandberg, Henrik},
  journal={IEEE Transactions on Automatic Control}, 
  title={A Bayesian Nash Equilibrium-Based Moving Target Defense Against Stealthy Sensor Attacks}, 
  year={2024},
  volume={69},
  number={3},
  pages={1659-1674},
  doi={10.1109/TAC.2023.3328754}
}

@article{Tung2024security,
  title={Security Allocation in Networked Control Systems under Stealthy Attacks},
  author={Anh Tung Nguyen and André M.H. Teixeira and Alexander Medvedev},
  journal={IEEE Transactions on Control of Network Systems},
  year={2024},
  doi={10.1109/TCNS.2024.3462546}
}

@incollection{Golbeck2013Chapter3,
  author    = {Golbeck, Jennifer},
  title     = {Network Structure and Measures}, 
  booktitle = {Analyzing the Social Web},
  publisher = {Newnes},
  year      = {2013},
  chapter   = {3},
  pages     = {25-44}
}

@book{trentelman1991dissipation,
  title={The dissipation inequality and the algebraic Riccati equation},
  author={Trentelman, Harry L and Willems, Jan C},
  year={1991},
  publisher={Springer}
}

@article{gupta2016dynamic,
	title={Dynamic games with asymmetric information and resource constrained players with applications to security of cyberphysical systems},
	author={Gupta, Abhishek and Langbort, C{\'e}dric and Ba{\c{s}}ar, Tamer},
	journal={IEEE Transactions on Control of Network Systems},
	volume={4},
	number={1},
	pages={71--81},
	year={2016},
	publisher={IEEE},
    doi={10.1109/TCNS.2016.2584183}
}

@article{miao2018hybrid,
	title={A hybrid stochastic game for secure control of cyber-physical systems},
	author={Miao, Fei and Zhu, Quanyan and Pajic, Miroslav and Pappas, George J},
	journal={Automatica},
	volume={93},
	pages={55--63},
	year={2018},
	publisher={Elsevier},
    doi={10.1016/j.automatica.2018.03.012}
}

@article{falliere2011w32,
  title={W32. stuxnet dossier},
  author={Falliere, Nicolas and Murchu, Liam O and Chien, Eric},
  journal={White paper, Symantec Corp., Security Response},
  volume={5},
  number={6},
  pages={29},
  year={2011}
}

@article{teixeira2015secure,
  title={A secure control framework for resource-limited adversaries},
  author={Teixeira, Andr{\'e} and Shames, Iman and Sandberg, Henrik and Johansson, Karl Henrik},
  journal={Automatica},
  volume={51},
  pages={135--148},
  year={2015},
  publisher={Elsevier},
  doi={10.1016/j.automatica.2014.10.067}
}

@article{nguyen2022single,
  title={A Single-Adversary-Single-Detector Zero-Sum Game in Networked Control Systems},
  author={Nguyen, Anh Tung and Teixeira, Andr{\'e} M H  and Medvedev, Alexander},
  journal={IFAC-PapersOnLine},
  volume={55},
  number={13},
  pages={49--54},
  year={2022},
  publisher={Elsevier},
  doi ={j.ifacol.2022.07.234}
}

@article{kshetri2017hacking,
  title={Hacking power grids: A current problem},
  author={Kshetri, Nir and Voas, Jeffrey},
  journal={Computer},
  volume={50},
  number={12},
  pages={91--95},
  year={2017},
  publisher={IEEE},
  doi={10.1109/MC.2017.4451203}
}

@inproceedings{nguyen2022zero,
  title={A zero-sum game framework for optimal sensor placement in uncertain networked control systems under cyber-attacks},
  author={Nguyen, Anh Tung and Anand, Sribalaji C and Teixeira, Andr{\'e} MH},
  booktitle={2022 IEEE 61st Conference on Decision and Control (CDC)},
  pages={6126--6133},
  year={2022},
  organization={IEEE},
  doi={10.1109/CDC51059.2022.9992468}
}

@phdthesis{tegling2018fundamental,
  title={Fundamental limitations of distributed feedback control in large-scale networks},
  author={Tegling, Emma},
  year={2018},
  school={KTH Royal Institute of Technology}
}

@misc{ieee14bus,
author = {UW-EE},
title = {IEEE 14-Bus Test Case},
month = August,
year = {1993},
url = {labs.ece.uw.edu/pstca/pf14/ieee14cdf.txt}
}

@book{bacsar1998dynamic,
  title={Dynamic noncooperative game theory},
  author={Ba{\c{s}}ar, Tamer and Olsder, Geert Jan},
  year={1998},
  publisher={SIAM}
}

@article{li2018false,
  title={False data injection attacks on networked control systems: A Stackelberg game analysis},
  author={Li, Yuzhe and Shi, Dawei and Chen, Tongwen},
  journal={IEEE Transactions on Automatic Control},
  volume={63},
  number={10},
  pages={3503--3509},
  year={2018},
  publisher={IEEE},
  doi={10.1109/TAC.2018.2798817}
}

@article{shukla2022robust,
  title={A Robust Stackelberg Game for Cyber-Security Investment in Networked Control Systems},
  author={Shukla, Pratishtha and An, Lu and Chakrabortty, Aranya and Duel-Hallen, Alexandra},
  journal={IEEE Transactions on Control Systems Technology},
  volume={31},
  number={2},
  pages={856--871},
  year={2022},
  publisher={IEEE},
  doi={10.1109/TCST.2022.3207671}
}

@article{yuan2019stackelberg,
  title={Stackelberg-game-based defense analysis against advanced persistent threats on cloud control system},
  author={Yuan, Huanhuan and Xia, Yuanqing and Zhang, Jinhui and Yang, Hongjiu and Mahmoud, Magdi S},
  journal={IEEE Transactions on Industrial Informatics},
  volume={16},
  number={3},
  pages={1571--1580},
  year={2019},
  publisher={IEEE},
  doi={10.1109/TII.2019.2925035}
}

@article{milovsevic2023strategic,
  title={Strategic Monitoring of Networked Systems with Heterogeneous Security Levels},
  author={Milo{\v{s}}evi{\'c}, Jezdimir and Dahan, Mathieu and Amin, Saurabh and Sandberg, Henrik},
  journal={IEEE Transactions on Control of Network Systems},
  year={2023},
  pages={1165--1176},
  volume={11},
  number={3},
  publisher={IEEE},
  doi={10.1109/TCNS.2023.3333392}
}

@article{zhang2021stealthy,
  title={Stealthy integrity attacks for a class of nonlinear cyber-physical systems},
  author={Zhang, Kangkang and Keliris, Christodoulos and Parisini, Thomas and Polycarpou, Marios M},
  journal={IEEE Transactions on Automatic Control},
  volume={67},
  number={12},
  pages={6723--6730},
  year={2021},
  publisher={IEEE},
  doi={10.1109/TAC.2021.3131656}
}

@article{park2019stealthy,
  title={Stealthy adversaries against uncertain cyber-physical systems: Threat of robust zero-dynamics attack},
  author={Park, Gyunghoon and Lee, Chanhwa and Shim, Hyungbo and Eun, Yongsoon and Johansson, Karl H},
  journal={IEEE Transactions on Automatic Control},
  volume={64},
  number={12},
  pages={4907--4919},
  year={2019},
  publisher={IEEE},
  doi={10.1109/TAC.2019.2903429}
}

@article{mo2013detecting,
  title={Detecting integrity attacks on SCADA systems},
  author={Mo, Yilin and Chabukswar, Rohan and Sinopoli, Bruno},
  journal={IEEE Transactions on Control Systems Technology},
  volume={22},
  number={4},
  pages={1396--1407},
  year={2013},
  publisher={IEEE},
  doi={10.1109/TCST.2013.2280899}
}

@article{naha2023quickest,
  title={Quickest physical watermarking-based detection of measurement replacement attacks in networked control systems},
  author={Naha, Arunava and Teixeira, Andr{\'e} and Ahl{\'e}n, Anders and Dey, Subhrakanti},
  journal={European Journal of Control},
  volume={71},
  pages={100804},
  year={2023},
  publisher={Elsevier},
  doi={10.1016/j.ejcon.2023.100804}
}

@inproceedings{teixeira2019optimal,
  title={Optimal stealthy attacks on actuators for strictly proper systems},
  author={Teixeira, Andr{\'e} M H},
  booktitle={2019 IEEE 58th Conference on Decision and Control (CDC)},
  pages={4385--4390},
  year={2019},
  organization={IEEE},
  doi={10.1109/CDC40024.2019.9029171}
}

@article{ferrari2020switching,
  title={A switching multiplicative watermarking scheme for detection of stealthy cyber-attacks},
  author={Ferrari, Riccardo M G and Teixeira, Andr{\'e} M H},
  journal={IEEE Transactions on Automatic Control},
  volume={66},
  number={6},
  pages={2558--2573},
  year={2020},
  publisher={IEEE},
  doi={10.1109/TAC.2020.3013850}
}

@inproceedings{gallo2021design,
  title={Design of multiplicative watermarking against covert attacks},
  author={Gallo, Alexander J and Anand, Sribalaji C and Teixeira, Andr{\'e} M H and Ferrari, Riccardo MG},
  booktitle={2021 60th IEEE Conference on Decision and Control (CDC)},
  pages={4176--4181},
  year={2021},
  organization={IEEE},
  doi={10.1109/CDC45484.2021.9683075}
}

@inproceedings{anand2022risk,
  title={Risk assessment and optimal allocation of security measures under stealthy false data injection attacks},
  author={Anand, Sribalaji C and Teixeira, Andr{\'e} M H and Ahl{\'e}n, Anders},
  booktitle={2022 IEEE Conference on Control Technology and Applications (CCTA)},
  pages={1347--1353},
  year={2022},
  organization={IEEE},
  doi={10.1109/CCTA49430.2022.9966025}
}

@inproceedings{li2023secure,
  title={Secure state estimation with asynchronous measurements against malicious measurement-data and time-stamp manipulation},
  author={Li, Zishuo and Nguyen, Anh Tung and Teixeira, Andr{\'e} MH and Mo, Yilin and Johansson, Karl H},
  booktitle={2023 62nd IEEE Conference on Decision and Control (CDC)},
  pages={7073--7080},
  year={2023},
  organization={IEEE},
  doi={10.1109/CDC49753.2023.10383571}
}

@article{li2024secure,
  title={Secure Distributed Dynamic State Estimation against Sparse Integrity Attack via Distributed Convex Optimization},
  author={Li, Zishuo and Mo, Yilin},
  journal={IEEE Transactions on Automatic Control},
  year={2024},
  pages={6089--6104},
  volume={69},
  number={9},
  publisher={IEEE},
  doi={10.1109/TAC.2024.3397158}
}

@article{nguyen2023optimal,
  title={Optimal detector placement in networked control systems under cyber-attacks with applications to power networks},
  author={Nguyen, Anh Tung and Anand, Sribalaji C and Teixeira, Andr{\'e} MH and Medvedev, Alexander},
  journal={IFAC-PapersOnLine},
  volume={56},
  number={2},
  pages={1820--1826},
  year={2023},
  publisher={Elsevier},
  doi={10.1016/j.ifacol.2023.10.1896}
}
